\documentclass[11pt,draftcls,onecolumn]{IEEEtran}

\usepackage{psfrag,amssymb,amsmath,pifont,cite,graphics,graphicx,epsfig,subfigure,amscd,helvet,multirow,enumerate}
\usepackage{lmodern}

\newtheorem{theorem}{Theorem}[section]
\newtheorem{lemma}[theorem]{Lemma}

\newtheorem{corollary}[theorem]{Corollary}
\newtheorem{definition}[theorem]{Definition}

\begin{document}
\title{Efficient Soft-Input Soft-Output Tree  Detection \\ Via an Improved Path Metric}

\author{\IEEEauthorblockN{ Jun Won Choi*, Byonghyo Shim,~\IEEEmembership{Senior Member,~IEEE,} \\ and Andrew C. Singer,~\IEEEmembership{Fellow,~IEEE}}
%\IEEEauthorblockA{\\ *University of Illinois at Urbana-Champaign\\
% Email: jwchoi@illinois.edu
%  }
%\thanks{Dedicated to the memory of Ralf Koetter, with deep respect and admiration.}
\thanks{J. W. Choi and A. C. Singer are with Dept. of Electrical and Computer Engineering, University of Illinois at Urbana-Champaign. B. Shim is with EECS Dept., Korea Univ., Seoul, Korea.}
}

%\and
%\IEEEauthorblockN{James Kirk\\ and Montgomery Scott}
%\IEEEauthorblockA{Starfleet Academy\\
%San Francisco, California 96678-2391\\
%Fax: (888) 555--1212}


%\author{
%
%\IEEEauthorblockN{Michael Shell\IEEEauthorrefmark{1},
%Telephone: (800) 555--1212, Fax: (888) 555--1212}
%\IEEEauthorblockA{\IEEEauthorrefmark{4}Tyrell Inc., 123 Replicant Street, Los Angeles, California 90210--4321}
%}


\maketitle

\begin{abstract}
  Tree detection techniques are often used to reduce the complexity of \emph{a posteriori probability} (APP) detection in high dimensional multi-antenna wireless communication systems. In this paper, we introduce an efficient soft-input soft-output tree detection algorithm that employs a new type of look-ahead path metric in the computation of its branch pruning (or sorting). While conventional path metrics depend only on symbols on a visited path, the new path metric accounts for unvisited parts of the tree in advance through an unconstrained linear estimator and adds a bias term that reflects the contribution of as-yet undecided symbols. By applying the linear estimate-based look-ahead path metric to an $M$-algorithm that selects the best $M$ paths for each level of the tree we develop a new soft-input soft-output tree detector, called an \emph{improved soft-input soft-output $M$-algorithm} (ISS-MA).  Based on an analysis of the probability of correct path loss, we show that  the improved path metric  offers substantial performance gain over the conventional path metric. We also demonstrate through simulations that the ISS-MA provides a better performance-complexity trade-off than existing soft-input soft-output detection algorithms.
\end{abstract}

\IEEEpeerreviewmaketitle

\section{Introduction}
\label{sec:introduction}
The relationship between the transmitted symbol and the received signal vector in many communication systems can be expressed in the form
\begin{equation}
\label{eq:start}
\mathbf{y}_{o} = \mathbf{H} \mathbf{x} + \mathbf{n}_{o},
\end{equation}
where $\mathbf{x}$ is the $N \times 1$ transmitted vector whose entries are
chosen from a finite symbol alphabet, $\mathbf{y}_{o}$ and $\mathbf{n}_{o}$
are the $L \times 1$ received signal and noise vectors, respectively, and
$\mathbf{H}$ is $L\times N$ channel matrix.
As a practical decoding scheme when a code constraint is imposed, \emph{iterative detection and decoding} (IDD) has been applied to various digital communication
systems including channel equalization \cite{singer}, multi-input multi-output (MIMO) detection \cite{haykin,hochwald,vikalo2004}, and multi-user detection \cite{wang_poor}.
Motivated by the \emph{turbo principle} \cite{berrou}, an IDD receiver  exchanges soft information between a symbol detector and a channel decoder to achieve performance close to the channel capacity.
%
%In fact, through the iterative IDD process, near-capacity performance in the MIMO system has been achieved \cite{hochwald}.
%Through the iterative IDD process, the receiver can achieve a near-capacity performance  \cite{hochwald}.
%
The symbol detector computes \emph{a posteriori} probabilities (APP)
on the bits comprising $\mathbf{x}$, using \emph{a priori} probabilities provided by the channel decoder  and the observation $\mathbf{y}_{o}$.
Then, the detector exchanges this soft information (so called \emph{extrinsic} information) with a soft-input soft-output decoder, such as the
max-log-MAP decoder \cite{logmapsova}.
In the sequel, we refer to such a detector as an \emph{APP detector}.
%The receiver aims to detect $\mathbf{x}$ based on
%
%In particular, \emph{a posteriori probability} (APP) detection is formulated as obtaining \emph{a posteriori} probabilities on the bits comprising $\mathbf{x}$ based on \emph{a priori} probabilities on them and the observed signal $\mathbf{y}$.
%%
%The APP detector is a key component of   \emph{iterative detection and decoding} (IDD) system, where a form of soft information is exchanged  between
%a soft-input soft-output (soft-input soft-output) symbol detector and soft-input soft-output channel decoder to achieve bit error rate (BER) improvement \cite{hochwald}.
%

Direct computation of the APP involves marginalization over all configurations of the vector $\mathbf{x}$, leading to exponential complexity in the system size (e.g., number of antenna elements in MIMO systems). As a means of approximately
 performing the APP detection at reduced complexity, \emph{tree detection techniques} have received much attention recently \cite{hochwald,vikalo2004,giannakis2006,studer,jalden2005,jong,milliner,hagenauer2007,barbero2008,larsson,wu}. (Refer to \cite{arul} for an overview of tree detection techniques.) The essence of these approaches is to  produce a set of promising symbol candidates via a tree search for estimating the APP over this reduced set.
Thus far, a variety of tree detection algorithms have been proposed.
In \cite{hochwald}, the \emph{sphere decoding algorithm} (SDA) \cite{pohst,hassibi1} with a fixed radius was used to find symbol candidates.
%
%This method, called list sphere decoding (LSD), finds a list of symbol candidates with high likelihood function through  a fixed-radius  search and
%computes approximated APP based on the candidate list.
 % and then the APP  is estimated by counting the symbol candidates in the list.
%
%However, it is difficult to choose an appropriate radius, which provides good complexity with small performance loss.
%
In \cite{vikalo2004}, \emph{a priori} information obtained from the channel decoder was exploited to improve the search efficiency of the SDA.
In \cite{giannakis2006},  a hard sphere decoder was employed to find a single \emph{maximum a posteriori} probability (MAP) symbol estimate maximizing $P\left(\mathbf{x}|\mathbf{y}_{o}\right)$ and a candidate list was generated by flipping bits in the MAP estimate.
In \cite{jalden2005}, the APPs of all bits in $\mathbf{x}$ are obtained simultaneously by modifying a bound tightening rule of a single sphere search. Additionally, a more sophisticated extension of this idea was introduced in \cite{studer}.
The computational complexity of these tree detection algorithms varies depending on the channel and noise realizations, and in the worst-case the search complexity is the same as that of exhaustive search.
 In order to limit the worst case complexity of the tree detection approach, \emph{fixed-complexity tree search} techniques \cite{jalden2007} have been proposed.
%
%While the complexity of these APP detectors depends on noise and channel realizations, there exist fixed-complexity APP detection algorithms.
%
For example, an $M$-algorithm was extended to soft-input soft-output detection in \cite{jong} and an intelligent candidate adding algorithm for improving efficiency of the $M$-algorithm was proposed in \cite{milliner}.
The stack algorithm was also exploited  for list generation in combination with soft augmentation of tail bits of stack elements \cite{hagenauer2007}.
Other fixed-complexity soft-input soft-output detection algorithms can be found in \cite{barbero2008,larsson,wu}.

The $M$-algorithm  \cite{jong,malg}, also known as \emph{$K$-best algorithm} in the MIMO detection literature  \cite{wong,guo},
selects only a finite set of the $M$ best candidates for each layer of the detection tree.
The $M$-algorithm is a practical candidate for soft-input soft-output detection
due to its inherent nature to facilitate parallel and pipelined processing \cite{guo}.
%
%Moreover, since the $M$-algorithm finds best multiple symbol candidates, it can naturally be extended to APP detection \cite{jong}.
%
%In addition, it has a fixed complexity and can be configured only through a single parameter $M$.
%
In spite of this benefit, the $M$-algorithm   suffers from a poor performance-complexity trade-off due to
the greedy nature of the algorithm. To be specific, the algorithm checks the validity of paths in the forward direction and never traverses back for reconsideration.
Once the correct path is rejected, it will never be selected again in subsequent selections,  resulting in wasteful search effort.
Moreover, these erroneous decisions often occur in early candidate selection stages
where   the accumulated path metric  considers only  a few symbol spans.
%This increases the chance of rejecting
%%
%promising paths, limiting performance-complexity trade-off of the $M$-algorithm.
%
One way to alleviate such error propagation is symbol detection ordering \cite{milliner,vblast}. By processing each layer in an appropriate order, the chances of errors propagating to the next stage can be reduced.
Nevertheless, error propagation severely limits the performance of the $M$-algorithm especially when the system size is large.

In this paper, we pursue an improvement of the performance-complexity trade-off  of soft-input soft-output $M$-algorithms.
%by designing a new path metric.
%
Towards this end, we propose a new path metric capturing the contribution of the entire symbol path.
%
% to reduce the chance of rejecting reliable paths.
%
%enhances the sorting process of the $M$-algorithm by employing a path metric deliberately designed to capture the contribution of the entire symbol path.
%
%
While the conventional path metric accounts for the contributions of symbols along the visited path only, the new path metric looks ahead to
the unvisited paths and estimates their contributions through a soft unconstrained linear symbol estimate.
  In fact, a \emph{bias term} reflecting the information from as-yet undecided symbols is incorporated into the conventional path metric for this purpose.
In order to distinguish this improved path metric from the conventional path metric and other look-ahead metrics, we henceforth refer to it as a \emph{linear estimate-based look-ahead (LE-LA) path metric}.
We apply the LE-LA path metric to the soft-input soft-output $M$-algorithm, introducing an \emph{improved soft-in soft-out $M$-algorithm} (ISS-MA).
By sorting paths based on the LE-LA path metric, the ISS-MA lessens the chance of rejecting the correct path from the candidate list and eventually improves
the detection performance  especially for systems of large dimension.
Indeed, from an analysis of the probability of correct path loss (CPL),
we show that the LE-LA path metric benefits the candidate selection process of the $M$-algorithm.
%an upper bound on the CPL probability of the ISS-MA is strictly smaller than that of the conventional path metric.
%

   The idea of using a look-ahead path metric has been explored in artificial intelligence search problems \cite{ai} and can also be found in
soft decoding of linear block codes \cite{hh,ekroot}.
In \cite{stojnic}, computationally efficient methods to obtain the bias term were investigated using semi-definite programming and $H^{\infty}$ estimation techniques.
While these approaches search for a \emph{deterministic} bias term (lower-bound of future cost) to guarantee the optimality of the sequential or depth-first search, our approach uses linear estimation to derive a bias term designed to improve candidate selection of the breadth-first search.
The key advantage of using a linear estimator is that \emph{a priori} information can be easily
incorporated into the bias term so that the look-ahead operation benefits from  the decoder output in each iteration.
It is also worth emphasizing the difference between the proposed path metric and Fano matric  \cite{fano}.
The Fano metric exploits the \emph{a posteriori} probability of  each path as its path metric.
For a binary symmetric channel, the Fano metric introduces a bias term proportional to the path length to penalize paths of short length.
The extension of the Fano metric to channels with memory or MIMO channels is not straightforward, since it involves marginalization over
the distribution of the undecided symbols.
Modification of the Fano metric is considered for equalization of intersymbol interference (ISI) channels in \cite{xiong} and for
multi-input multi-output detection in  \cite{arul}.
As a means to improve path metric of the SDA, the idea of probabilistic pruning was introduced in \cite{gowaikar,shim,ho,wanlun2006}.
In \cite{hagenauer2007},  the probability density of an observed signal estimated from a separate tree search is used as a bias term.
While these approaches assign an equal bias term for paths of the same length, the ISS-MA provides a distinct bias term for each path in the tree, allowing for the application of a breadth-first search such as the $M$-algorithm.
As such, our path metric can be readily combined with any type of tree-based soft-input soft-output detector.

The rest of this paper is organized as follows.
In Section \ref{sec:problem},  we briefly review the IDD system and the tree detection algorithm.
In Section \ref{sec:iss-ma}, we present the LE-LA path metric along with its efficient computation.
We also describe the application of the LE-LA path metric to the soft-input soft-output $M$-algorithm.
In Section \ref{sec:performance}, we present the performance analysis of the ISS-MA.
In Section \ref{sec:simulation}, we provide simulation results and conclude in Section \ref{sec:conclusion}.

We briefly summarize the notation used in this paper.
Uppercase and lowercase letters written in boldface denote matrices and vectors, respectively.
%Note that all vectors denote column vector.
The superscripts $(\cdot)^{T}$ and $(\cdot)^{H}$ denote  transpose and  conjugate transpose (hermitian operator), respectively. $\| \cdot\|^{2}$ denotes an $L_2$-norm square of a vector and
${\rm diag}\left(\cdot\right)$ is a diagonal matrix that has elements on the main diagonal.
%
%$\Re\left(x\right)$ and $\Im\left(x\right)$ denote the real and imaginary parts of $x$, respectively.
%
$\mathbf{0}_{M \times N}$ and $\mathbf{1}_{M \times N}$  are $M \times N$ matrix whose entries are all ones or zeros, respectively. The subscript is omitted if there is no risk of confusion.
%
%$\mathcal{X}_{k}^{2}$ denotes a Chi-square distribution with $k$ degrees of freedom (DOF).
%
%$F_{\chi} \left( \cdot \,;\, k \right)$ and $F_{\chi}^{-1} \left(\cdot \, ; \, k\right)$ are the cumulative density function (CDF) and the inverse CDF of the $\chi^2$-random variable with $k$ DOF, respectively. $Q(x)$ denotes the $Q$-function defined as $Q(x) = \int_{x}^{\infty}\frac{1}{\sqrt{2\pi}} {\rm exp}\left(-x^{2}/2\right)$.
%
$\mathcal{CN}(m,\sigma^{2})$ denotes a circular symmetric complex Gaussian density with mean $m$ and variance $\sigma^{2}$.
$E_x[\cdot]$ denotes expectation over the random variable $x$.
${\rm Cov} \left({\mathbf{x},\mathbf{y}}\right)$ denotes $E\left[\mathbf{x}\mathbf{y}^{H}\right] - E[\mathbf{x}]E[\mathbf{y}^{H}]$.
For a hermitian matrix $\mathbf{A}$, $\mathbf{A}\succeq 0$ (or $\mathbf{A}\succ 0 $) means that $\mathbf{A}$ is semi-positive definite (or positive definite).
$Pr(A)$ means probability of the event $A$.
$f_{x_1, x_2, \cdots, x_n}(a_1, a_2, \cdots, a_n)$ denotes a joint probability density function (PDF) for the random variables $x_1, x_2, \cdots, x_n$.
%

%
%$\mathcal{CN}(m,\sigma^{2})$ denotes a complex Gaussian with the mean $m$ and variance $\sigma^{2}$.

\section{Problem Description}
\label{sec:problem}
In this section, we briefly review
%we describe the IDD system, which perform
%symbol detection and channel decoding in an iterative fashion.
%%
%In the IDD system, the symbol detector  uses the soft information
%from the channel decoder to produce additional symbol information, called extrinsic information.
 the IDD framework and then introduce the tree detection algorithms.

\subsection{Iterative Detection and Decoding (IDD)}
\label{subsec:IDD}

In a transmitter,
a rate
$R_{c}$ channel encoder is used to convert a sequence of independent identically distributed (i.i.d.) binary information bits $\{b_{i}\}$  to an encoded sequence $\{c_{i}\}$.
%
% For convenience, we let the coded bit $c_{i}$ have the value of $\pm 1$.
%
 The bit sequence $\{c_{i}\}$ is permuted using a random interleaver $\prod$ and  then
  mapped into a symbol vector using a $2^Q$-ary quadrature amplitude modulation (QAM) symbol alphabet.
We label the interleaved bits associated with the $k$th symbol $x_{k}$ by $\overline{c}_{k,1}, \cdots, \overline{c}_{k,Q}$.
Due to the interleaver, we assume that these interleaved bits are mutually uncorrelated.

In the system model (\ref{eq:start}),
% $\mathbf{y} = \mathbf{H} \mathbf{x} + \mathbf{n}$,
%
%%%%%%%%%%%%%%%%%%%%%%%%%%%%%%%%%%%%%%%%%%%%%%%%%%%%%%%%%%%%%%%%%%%%%%%%%%%%%%%%%%%%%%%%%%%%%%%%%%%%
%\begin{equation}
%\label{eq:start}
%\mathbf{y} = \mathbf{H} \mathbf{x} + \mathbf{n},
%\end{equation}
%%%%%%%%%%%%%%%%%%%%%%%%%%%%%%%%%%%%%%%%%%%%%%%%%%%%%%%%%%%%%%%%%%%%%%%%%%%%%%%%%%%%%%%%%%%%%%%%%%%%
$\mathbf{y}_{o}$ and $\mathbf{n}_{o}$
are the $L \times 1$ received signal and noise vectors, respectively.
Each entry of the $N \times 1$ symbol vector $\mathbf{x}$ is drawn from a finite alphabet
\begin{align} \label{eq:qamsig}
\mathcal{F} =& \Big\{x_{r} + j x_{i}  |\;\; x_{r}, x_{i} \in
\Big\{ \frac{-2^{Q/2}+1}{P},\frac{-2^{Q/2}+3}{P},
\cdots,   \frac{2^{Q/2}-3}{P},
\frac{2^{Q/2}-1}{P} \Big\} \Big\},
\end{align}
%%%%%%%%%%%%%%%%%%%%%%%%%%%%%%%%%%%%%%%%%%%
where $P$ is chosen to satisfy the normalization condition $E\left[|x_{k}|^{2}\right] = 1$. For example, $P = \sqrt{10}$ for $16$-QAM and $P = \sqrt{42}$ for $64$-QAM modulation, respectively.

\begin{figure*}[t]
\centerline{\psfig{figure=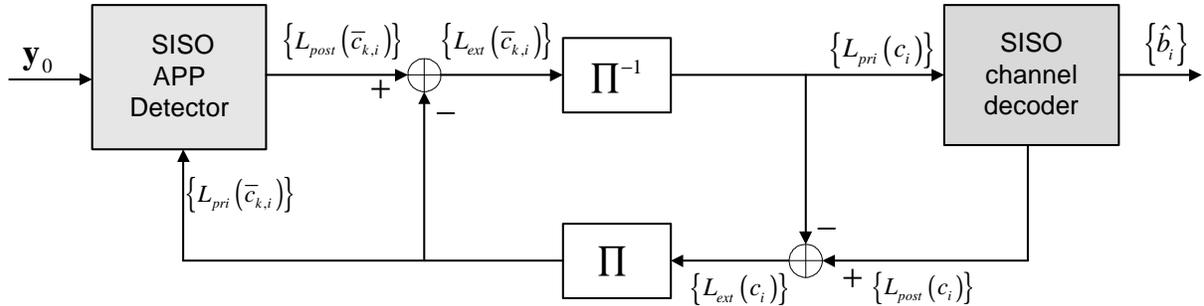,width=160mm}}
\caption{Block diagram of the IDD system. } \label{fig:idd}
\end{figure*}
Fig.~\ref{fig:idd} depicts the basic structure of an IDD system.
The receiver consists of two main blocks; the APP detector and the channel decoder.
The APP detector generates the \emph{a posteriori} log-likelihood ratio (LLR) of $\overline{c}_{k,i}$
using
 the observation $\mathbf{y}_{o}$ and \emph{a priori} information  delivered from the channel decoder.
The \emph{a posteriori} LLR  is defined as
%%%%%%%%%%%%%%%%%%%%%%%%%%%%%%%%%%%%%%%%%%%%%%%%%%%%%%%%%%%%%%%%%%%%%%%%%%%%%%%%%%%%%%%%%%%%%%%%%%%%
\begin{equation} \label{eq:softmap}
L_{\rm post}\left(\overline{c}_{k,i}\right) = \ln \frac{Pr(\overline{c}_{k,i} = +1 | \mathbf{y}_{o})}{Pr(\overline{c}_{k,i}=-1 | \mathbf{y}_{o})},
\end{equation}
where we take $\overline{c}_{k,i} \in \{-1,1\}$ rather than $\{0,1\}$ by convention.
%%%%%%%%%%%%%%%%%%%%%%%%%%%%%%%%%%%%%%%%%%%%%%%%%%%%%%%%%%%%%%%%%%%%%%%%%%%%%%%%%%%%%%%%%%%%%%%%%%%
%where $L_{\rm post}\left(\overline{c}_{k,i}\right)$ is called \emph{a posteriori}  log likelihood ratio (LLR).
%This LLR estimation is referred to as  soft-output MAP estimation. In this paper, the ``MAP decoding" implies
%this soft bit-level decoding.
%
%
With the standard noise model $\mathbf{n}_{o} \sim \mathcal{CN}\left(0, \sigma_{n}^{2}\mathbf{I}\right)$, (\ref{eq:softmap}) can be rewritten \cite{hochwald}
%
%%%%%%%%%%%%%%%%%%%%%%%%%%%%%%%%%%%%%%%%%%%%%%%%%%%%%%%%%%%%%%%%%%%%%%%%%%%%%%%%%%%%%%%%%%%%%%%%%%%
\begin{equation} \label{eq:mapprom}
L_{\rm post}\left(\overline{c}_{k,i}\right) = \ln \frac{\sum_{\mathbf{x}\in X_{k,i}^{+1}}\exp\left(\psi\left(\mathbf{x}\right)\right)}{\sum_{\mathbf{x}\in X_{k,i}^{-1}}\exp\left(\psi\left(\mathbf{x}\right)\right)},
\end{equation}
%%%%%%%%%%%%%%%%%%%%%%%%%%%%%%%%%%%%%%%%%%%%%%%%%%%%%%%%%%%%%%%%%%%%%%%%%%%%%%%%%%%%%%%%%%%%%%%%%%%
where
%%%%%%%%%%%%%%%%%%%%%%%%%%%%%%%%%%%%%%%%%%%%%%%%%%%%%%%%%%%%%%%%%%%%%%%%%%%%%%%%%%%%%%%%%%%%%%%%%%%
\begin{align}
\psi(\mathbf{x}) =&  -\frac{1}{\sigma_{n}^{2}}\left\|\mathbf{y}_{o} - \mathbf{H}\mathbf{x} \right\|^{2} +
\sum_{i=1}^{N}\sum_{j=1}^{Q}\ln  Pr\left(\overline{c}_{i,j}\right), \nonumber \\
Pr\left(\overline{c}_{i,j}\right) =&  \frac{1}{2} \left(1 + \overline{c}_{i,j} \tanh\left(\frac{L_{\rm pri}(\overline{c}_{i,j})}{2}\right) \right).
%\sum_{k',i'}\bar{c}_{k',i'}\frac{L_{\rm pri}(\bar{c}_{k',i'})}{2}, \nonumber
\end{align}
%%%%%%%%%%%%%%%%%%%%%%%%%%%%%%%%%%%%%%%%%%%%%%%%%%%%%%%%%%%%%%%%%%%%%%%%%%%%%%%%%%%%%%%%%%%%%%%%%%%
%\begin{align}
%L_{\rm post}\left(k,i\right) =
%%=& \ln \frac{Pr\left(\bar{c}_{k,i} = +1|\mathbf{y} \right)}{Pr\left(\bar{c}_{k,i} = -1|\mathbf{y} \right)} \\ \label{eq:map2}
%   \underbrace{\ln \frac{\sum_{\mathbf{x}\in X_{k,i}^{+1}}Pr\left(\mathbf{y}|\mathbf{x}\right)\prod_{k',i' \in \overline{C}_{k,i}}Pr(\bar{c}_{k',i'})}{\sum_{\mathbf{x}\in X_{k,i}^{-1}}Pr\left(\mathbf{y}|\mathbf{x}\right)\prod_{k',i' \in \overline{C}_{k,i}}Pr(\bar{c}_{k',i'})}}_{L_{\rm ext}(\bar{c}_{k,i})} -
%   \underbrace{\ln \frac{Pr\left(\bar{c}_{k,i} = +1\right)}{Pr\left(\bar{c}_{k,i} = -1\right)}}_{L_{\rm pri}(\bar{c}_{k,i})},
%\end{align}
The set $X_{k,i}^{+ 1}$ is the set of all configurations of the vector $\mathbf{x}$ satisfying $\overline{c}_{k,i} = + 1$ ($X_{k,i}^{-1}$ is defined similarly), and
$L_{\rm pri}(\overline{c}_{k,i})$ is the \emph{a priori}
LLR defined as $L_{\rm pri}(\overline{c}_{k,i}) = \ln Pr(\overline{c}_{k,i}=+1) - \ln Pr(\overline{c}_{k,i}=-1)$.
%
%Since the first term in (\ref{eq:map2}) is the \emph{a priori} LLR $L_{\rm pri}(\bar{c}_{k,i})$ which is given to the symbol detector as an input, the MAP detector computes
%the extrinsic LLR as
%%
%\begin{equation} \label{eq:mapprom}
%L_{\rm ext}(\bar{c}_{k,i}) = \ln \frac{\sum_{\mathbf{x}\in X_{k,i}^{+1}}\exp\left( -\frac{1}{\sigma_{n}^{2}}\left\|\mathbf{y} - \mathbf{H}\mathbf{x} \right\|^{2} +  \sum_{k',i' \in \overline{C}_{k,i}}\bar{c}_{k',i'}\frac{L_{\rm pri}(\bar{c}_{k',i'})}{2}\right)}{\sum_{\mathbf{x}\in X_{k,i}^{-1}}\exp\left( -\frac{1}{\sigma_{n}^{2}}\left\|\mathbf{y} - \mathbf{H}\mathbf{x} \right\|^{2} +  \sum_{k',i' \in \overline{C}_{k,i}}\bar{c}_{k',i'}\frac{L_{\rm pri}(\bar{c}_{k',i'})}{2}\right)}.
%\end{equation}
%
Once $L_{\rm post}(\overline{c}_{k,i})$ is computed, the extrinsic LLR is obtained from $L_{\rm ext}(\overline{c}_{k,i}) = L_{\rm post}(\overline{c}_{k,i}) - L_{\rm pri}(\overline{c}_{k,i})$. These extrinsic LLRs are de-interleaved and then
 delivered
to the channel decoder.
The channel decoder computes the extrinsic LLR for the coded bits $\{c_{i}\}$ and feeds them back to the APP detector.
%
%Then, \emph{a priori} LLRs are returned from the channel decoder and the APP detection is performed for next iteration.
%
These operations are repeated until a suitably chosen convergence criterion is achieved \cite{hochwald}.

\subsection{Soft-input Soft-output Tree Detection}
\label{subsec:siso}

The direct computation of the \emph{a posteriori} LLR in (\ref{eq:mapprom}) involves  marginalization over $2^{NQ}$ symbol candidates, which easily  becomes
infeasible for large systems employing high order modulations.
A tree detection algorithm addresses this problem by searching a small set of promising symbol candidates over which
\emph{a posteriori} LLRs are estimated.
Specifically, a small number of symbol vectors  with large $\psi(\mathbf{x})$, equivalently, small $-\sigma_{n}^{2}\psi(\mathbf{x})$, are sought.
In the sequel, we refer to  $d_{\rm APP}(\mathbf{x}) = -\sigma_{n}^{2}\psi(\mathbf{x})$ as a \emph{cost metric} for tree detection,
where $-\sigma_{n}^{2}$ is a scaling factor.
The goal of the tree detection algorithm is to find symbol vectors of small cost metric, and the best (minimum) among them corresponds to the
 \emph{maximum a posteriori} (MAP) solution (denoted by $\mathbf{x}_{\rm MAP}$).

The tree detection algorithm relies on a tree representation of the search space spanned by $\mathbf{x} = (x_1, \cdots, x_{N}) \in \mathcal{F}^{N}$.
% TSA efficiently finds the solution of (\ref{eq:mlprom3}) or (\ref{eq:logmapprom}) using smart branch pruning principle.
%
%
%Consider the combinations of symbols $x_{1}, \cdots, x_{N}$ comprising the vector $\mathbf{x}$.
Tree construction is performed from the root node as follows. First, representing the symbol realization for $x_{N}$, we  extend $2^Q$ branches from  the root (recall that we assume $2^Q$-ary QAM modulation).
For each such branch,  $2^Q$ child branches are extended for the possible realization of the next symbol $x_{N-1}$. These branch extensions are repeated until all branches corresponding to $x_{N}, \cdots, x_{1}$  are generated.
 This yields a tree of the depth $N$, where each ``complete" path  from the root to a leaf corresponds to a realization of $\mathbf{x}$.
In order to find the complete paths of small cost metric, the tree detection algorithm searches the tree using a systematic node visiting rule.
For notational simplicity, we henceforth denote a path associated with a set of symbols $x_{i}, \cdots, x_{j}, (i < j)$  by a column vector  $\mathbf{x}_{i}^{j} = \left[x_{i}, \cdots, x_{j}\right]^{T}$.
%
%We will use this notation for both path and vector without any confusion.
%
Also, we call a level of tree associated with the symbol $x_{i}$ ``the $i$th level" (e.g., the bottom level associated with $x_{1}$ is the first level).
For details on tree construction, see \cite{hassibi1}.
%An example of tree constructed for $M=2$ is illustrated in Fig. 1.

For a systematic search of symbol candidates, a \emph{path metric} is assigned to each path $\mathbf{x}_{i}^{N}$.
Towards this end, we perform a QR decomposition of $\mathbf{H}$ as
\begin{equation}
\mathbf{H} = \mathbf{Q} \begin{bmatrix} \mathbf{R} \\ \mathbf{0}\end{bmatrix} =  \left[\mathbf{Q}_{1}\;\; \mathbf{Q}_{2}\right] \begin{bmatrix} \mathbf{R} \\ \mathbf{0}\end{bmatrix},
\end{equation}
where  $\mathbf{R}$ has an $N\times N$ upper-triangular matrix whose diagonals are non-negative  and
$\mathbf{Q}$ is an $L\times N$ matrix satisfying $\mathbf{Q}^{H}\mathbf{Q} = \mathbf{I}$.
%
%
%
%The TSAs employs the systematic strategy to find the complete paths that have smallest distance metric .
%%
%For this purpose, for each node of the tree $X_{i}^{N}$, we
%assign a path metric $p^{T}\left(X_{i}^{N}\right)$.
%%
%The path metric is mostly chosen from probabilistic distribution of the path $X_{i}^{N}$ \cite{fano,xiong} or from a deterministic lower bound of
%$d(\mathbf{x})$  \cite{hassibi1,damen}.
%In this paper, we focus on deterministic approach that can be obtained by a QR decomposition of $\mathbf{H}$.
%%
%Suppose that $\mathbf{H}$ is decomposed to
%
Using the invariance of the norm to unitary transformations, we can define the cost metric $d_{\rm APP}(\mathbf{x})$ as
%
%%%%%%%%%%%%%%%%%%%%%%%%%%%%%%%%%%%%%%%%%%%%%%%%%%%%%%%%%%%%%%%%%%%%%%%%%%%%%%%%%%%%%%%%%%%%%%%%%%%
\begin{align}
\label{eq:afterqr1}
d_{\rm APP}(\mathbf{x}) = -\sigma_{n}^{2}\psi(\mathbf{x})
%\left\|\mathbf{y}_{o} - \mathbf{H}\mathbf{x} \right\|^{2} - \sigma_{n}^{2}  \sum_{i=1}^{N}\sum_{j=1}^{Q} \ln Pr\left(\overline{c}_{i,j}\right)
%\\
 = & \left\|\mathbf{y} - \mathbf{R}\mathbf{x} \right\|^{2} - \sigma_{n}^{2}  \sum_{i=1}^{N}\sum_{j=1}^{Q} \ln Pr\left(\overline{c}_{i,j}\right) +C \\ \label{eq:afterqr4} =&
\sum_{i=1}^{N} b(\mathbf{x}_{i}^{N}) + C,
\end{align}
%%%%%%%%%%%%%%%%%%%%%%%%%%%%%%%%%%%%%%%%%%%%%%%%%%%%%%%%%%%%%%%%%%%%%%%%%%%%%%%%%%%%%%%%%%%%%%%%%%%
where
%%%%%%%%%%%%%%%%%%%%%%%%%%%%%%%%%%%%%%%%%%%%%%%%%%%%%%%%%%%%%%%%%%%%%%%%%%%%%%%%%%%%%%%%%%%%%%%%%%%
%\begin{equation} \label{eq:branch}
$b(\mathbf{x}_{i}^{N}) =\left|y'_{i} - \sum_{j=i}^{N}r_{i,j}x_{j} \right|^{2} - \sigma_{n}^{2}  \sum_{i=1}^{Q} \ln Pr\left(\overline{c}_{k,i}\right)$,
%\end{equation}
%%%%%%%%%%%%%%%%%%%%%%%%%%%%%%%%%%%%%%%%%%%%%%%%%%%%%%%%%%%%%%%%%%%%%%%%%%%%%%%%%%%%%%%%%%%%%%%%%%%
and $\mathbf{y} = \left[y_{1}, \cdots, y_{N} \right]^{T} = \mathbf{Q}_{1}^{H}\mathbf{y}_{o}$ and $C = \left\|\mathbf{Q}_{2}^{H} \mathbf{y}_{o} \right\|^{2}$.
%The term $C$ can be neglected due to the irrelevance to $\mathbf{x}$.
%
%Note that in the absence of \emph{a priori} LLRs, $d(X_{i}^{N})$ reduces to $\left|y'_{i} - \sum_{j=i}^{N}r_{i,j}x_{j} \right|^{2}$.
%
The path metric associated with the path $\mathbf{x}_{k}^{N}$ can be defined as a partial sum in the cost metric \cite{hassibi1}
%%%%%%%%%%%%%%%%%%%%%%%%%%%%%%%%%%%%%%%%%%%%%%%%%%%%%%%%%%%%%%%%%%%%%%%%%%%%%%%%%%%%%%%%%%%%%%%%%%%
\begin{align} \label{eq:lowerb}
\gamma^{(c)}\left(\mathbf{x}_{k}^{N}\right) = \sum_{i=k}^{N} b(\mathbf{x}_{i}^{N}).
\end{align}
%%%%%%%%%%%%%%%%%%%%%%%%%%%%%%%%%%%%%%%%%%%%%%%%%%%%%%%%%%%%%%%%%%%%%%%%%%%%%%%%%%%%%%%%%%%%%%%%%%%
%
%Note that this path metric is obtained as a partial sum of cost metric in (\ref{eq:afterqr4}).
%
Whenever a new node is visited, the term $b(\mathbf{x}_{i}^{N})$, referred to as a branch metric, is added to the path metric of the parent node.
Since  the branch metric is non-negative for all $i$,  the path metric $\gamma^{(c)}\left(\mathbf{x}_{k}^{N}\right)$ becomes a lower bound of the cost metric $d_{\rm APP}(\mathbf{x})$.
Using $\gamma^{(c)}\left(\mathbf{x}_{k}^{N}\right)$, the tree detection algorithm compares the reliability of distinct paths and chooses the surviving paths.
Since the path metric is determined by the visited path, we henceforth denote $\gamma^{(c)}\left(\mathbf{x}_{k}^{N}\right)$ as a \emph{causal path metric}.

According to a predefined node visiting rule \cite{arul}, the tree detection algorithm attempts to find the \emph{complete paths} associated with smallest cost metric.
Denoting the set of the corresponding symbol candidates as $\mathcal{L}$, an approximate APP can be expressed as
 \begin{equation} \label{eq:mapprom_dasi}
L_{\rm post}\left(\overline{c}_{k,i}\right) \approx \ln \frac{\sum_{\mathbf{x}\in \mathcal{L} \cap X_{k,i}^{+1}}\exp\left(\psi\left(\mathbf{x}\right)\right)}{\sum_{\mathbf{x}\in \mathcal{L} \cap X_{k,i}^{-1}}\exp\left(\psi\left(\mathbf{x}\right)\right)}.
\end{equation}
Further simplification can be achieved using max-log approximation \cite{logmapsova,proakis}
\begin{equation} \label{eq:logmapprom_dasi}
L_{\rm post}(\overline{c}_{k,i}) \approx \left(\max_{\mathbf{x}\in \mathcal{L} \cap X_{k,i}^{+1}} \psi\left(\mathbf{x}\right) -  \max_{\mathbf{x}\in \mathcal{L} \cap X_{k,i}^{-1}} \psi\left(\mathbf{x}\right)\right).
\end{equation}
Since $\mathcal{L}$ does not span whole symbol space, either $\mathcal{L} \cap X_{k,i}^{+1}$ or $\mathcal{L} \cap X_{k,i}^{-1}$ might be empty for some values of $k$.
If this case happens, the magnitude of $L_{\rm post}(\overline{c}_{k,i})$ is set to infinity, causing a bias in LLR values.
One way to cope with this event is to clip the magnitude of $L_{\rm post}(\overline{c}_{k,i})$
to a constant value (e.g., $\pm 8$) \cite{hochwald}.
%is clipped to a constant \cite{} or

%See  for comprehensive review on the tree detection algorithms.
%Since this path metric is used in conventional $M$-algorithm, we refer to $p^{(T)}\left(\mathbf{x}_{k}^{N}\right)$ as
%

%
%Besides the path metric, tree detection algorithms are characterized by node visiting order and bound tightening rule.
%
%One example of principles that can be employed to reduce the search complexity is  so called
%%
%\emph{branch and bound (BB)} principle \cite{arul}.
%%
%%
%In conclusion, the candidates in the search space $\mathcal{F}^{N}$ can be efficiently searched using
%the tree representation and path metric assigned to each node.
%%In the sequel,  this path metric is referred to as a \emph{conservative path metric} since it is the strict lower-bound of the distance metric and hence conservatively chosen.

%%
%For efficient search, the tree search algorithm employs a tree pruning criterion such as
%\emph{branch and bound (BB)} principle. We refer to \cite{arul} for the details.
%%
%%
%The complexity of the tree search algorithm is determined by how to choose path metric, pruning criterion, visiting order, and so on.
%%In the sequel,  this path metric is referred to as a \emph{conservative path metric} since it is the strict lower-bound of the distance metric and hence conservatively chosen.
%

%%%%%%%%%%%%%%%%%%%%%%%%%%%%%%%%%%%%%%%%%%%%%%%%%%%%%%%%%%%%%%%%%%%%%%%%%%%%%%%%%%%%%%%%%%%%%%%%%%%
%%%%%%%%%%%%%%%%%%%%%%%%%%%%%%%%%%%%%%%%%%%%%%%%%%%%%%%%%%%%%%%%%%%%%%%%%%%%%%%%%%%%%%%%%%%%%%%%%%%
\section{Improved Soft-input Soft-output $M$-algorithm (ISS-MA)}
\label{sec:iss-ma}
%%%%%%%%%%%%%%%%%%%%%%%%%%%%%%%%%%%%%%%%%%%%%%%%%%%%%%%%%%%%%%%%%%%%%%%%%%%%%%%%%%%%%%%%%%%%%%%%%%%
%%%%%%%%%%%%%%%%%%%%%%%%%%%%%%%%%%%%%%%%%%%%%%%%%%%%%%%%%%%%%%%%%%%%%%%%%%%%%%%%%%%%%%%%%%%%%%%%%%%
In this section, we present the ISS-MA, which improves candidate selection process of the soft-input soft-output tree detection algorithms. We first describe a genie-aided path metric that motivates our work and then  introduce the new path metric that accounts for the information on
 unvisited paths.
 We also discuss an efficient way to compute the new path metric.
%

%%%%%%%%%%%%%%%%%%%%%%%%%%%%%%%%%%%%%%%%%%%%%%%%%%%%%%%%%%%%%%%%%%%%%%%%%%%%%%%%%%%%%%%%%%%%%%%%%%%
%%%%%%%%%%%%%%%%%%%%%%%%%%%%%%%%%%%%%%%%%%%%%%%%%%%%%%%%%%%%%%%%%%%%%%%%%%%%%%%%%%%%%%%%%%%%%%%%%%%
\subsection{Motivation}
\label{subsec:motivation}
%%%%%%%%%%%%%%%%%%%%%%%%%%%%%%%%%%%%%%%%%%%%%%%%%%%%%%%%%%%%%%%%%%%%%%%%%%%%%%%%%%%%%%%%%%%%%%%%%%%
%%%%%%%%%%%%%%%%%%%%%%%%%%%%%%%%%%%%%%%%%%%%%%%%%%%%%%%%%%%%%%%%%%%%%%%%%%%%%%%%%%%%%%%%%%%%%%%%%%%

We begin our discussion with the following path metric:
\begin{definition}
A \emph{genie-aided path metric} $\gamma^{(g)}\left(\mathbf{x}_{k}^{N}\right)$ is defined as
%%%%%%%%%%%%%%%%%%%%%%%%%%%%%%%%%%%%%%%%%%%%%%%%%%%%%%%%%%%%%%%%%%%%%%%%%%%%%%%%%%%%%%%%%%%%%%%%%%%
\begin{equation} \label{eq:genie}
\gamma^{(g)}\left(\mathbf{x}_{k}^{N}\right)  = \gamma^{(c)}\left(\mathbf{x}_{k}^{N}\right) + \underbrace{\min_{\mathbf{x}_{1}^{k-1}}\left(\sum_{i=1}^{k-1} b\left(\mathbf{x}_{i}^{N}\right)\right) }_{\mbox{bias term}}.
\end{equation}
%%%%%%%%%%%%%%%%%%%%%%%%%%%%%%%%%%%%%%%%%%%%%%%%%%%%%%%%%%%%%%%%%%%%%%%%%%%%%%%%%%%%%%%%%%%%%%%%%%%
\end{definition}
The genie-aided path metric is obtained by minimizing the sum of $b(\mathbf{x}_{i}^{N}) (1\leq  i \leq k-1)$  over all combinations of undecided symbols $\mathbf{x}_{1}^{k-1}$.
 This minimal term, which can be considered as a bias term, is added into the causal path metric.
%
%The bias term accounting for the contributions of unvisited paths is unavailable and thus dropped in the causal path metric.
%
%Since we suppose that the bias term is obtained without incurring any computations, we term it \emph{genie-aided}.
%
%\begin{lemma}
The genie-aided path metric can be used in the $M$-algorithm so that the best $M$ candidates with the smallest genie-aided path metric
are selected in each tree level.
It is easy to show that the $M$-algorithm employing the genie-aided path metric finds the closest (best) path
%only with minimal number of node visitation (i.e., with $M=1$).
with probability one  (even for $M=1$).
%\end{lemma}
This can be readily shown since the genie-aided path metric provides the smallest cost metric among all tail paths.
Note that a similar path metric appeared in \cite{stojnic}.
%
%Note that using the genie-aided path metric, the best path can be found with probability one for an arbitrary value of $M$.
\begin{theorem}
\label{theorem:map}
Given the actual transmitted symbol vector $\widetilde{\mathbf{x}}_{k}^{N}$ (i.e., $\mathbf{x}_{k}^{N} = \widetilde{\mathbf{x}}_{k}^{N}$), the bias term of the genie-aided path metric is
\begin{equation} \label{eq:geniegenie}
\min_{\mathbf{x}_{1}^{k-1}}\left(\sum_{i=1}^{k-1} b\left(\mathbf{x}_{i}^{N}\right)\right) = \sum_{i=1}^{k-1} b\left(\mathbf{x}_{i}^{N}\right) \bigg|_{\mathbf{x}_{1}^{k-1} = \check{\mathbf{x}}_{1}^{k-1}},
\end{equation}
where the minimizer $\check{\mathbf{x}}_{1}^{k-1}$ is the MAP estimate of $\mathbf{x}_{1}^{k-1}$, i.e.,
%%%%%%%%%%%%%%%%%%%%%%%%%%%%%%%%%%%%%%%%%%%%%%%%%%%%%%%%%%%%%%%%%%%%%%%%%%%%%%%%%%%%%%%%%%%%%%%%%%%
%%%%%%%%%%%%%%%%%%%%%%%%%%%%%%%%%%%%%%%%%%%%%%%%%%%%%%%%%%%%%%%%%%%%%%%%%%%%%%%%%%%%%%%%%%%%%%%%%%%
%%%%%%%%%%%%%%%%%%%%%%%%%%%%%%%%%%%%%%%%%%%%%%%%%%%%%%%%%%%%%%%%%%%%%%%%%%%%%%%%%%%%%%%%%%%%%%%%%%%
\begin{equation} \label{eq:maxmax}
\check{\mathbf{x}}_{1}^{k-1} = \arg \max_{\mathbf{x}_{1}^{k-1}} \ln Pr\left(\mathbf{x}_{1}^{k-1} \Big| \mathbf{y}, \mathbf{x}_{k}^{N} = \widetilde{\mathbf{x}}_{k}^{N}\right).
\end{equation}
%%%%%%%%%%%%%%%%%%%%%%%%%%%%%%%%%%%%%%%%%%%%%%%%%%%%%%%%%%%%%%%%%%%%%%%%%%%%%%%%%%%%%%%%%%%%%%%%%%%
\end{theorem}
\begin{proof}
See Appendix \ref{app:map}.
\end{proof}
\emph{Theorem} \ref{theorem:map} implies that the bias term of the genie-aided path metric is obtained by computing $\sum_{i=1}^{k-1} b\left(\mathbf{x}_{i}^{N}\right)$ using the
MAP estimate of $\mathbf{x}_{1}^{k-1}$.
This MAP estimate is derived under the condition that the path associated with the actual transmitted symbols, $\widetilde{\mathbf{x}}_{k}^{N}$ is given.
%, which is obtained when transmitted data $\widetilde{\mathbf{x}}_{k}^{N}$ is given.
%
%
% the minimizer $\check{\mathbf{x}}_{1}^{k-1}$ corresponds to the MAP estimate of $\mathbf{x}_{1}^{k-1}$ given the transmitted symbol vector $\widetilde{\mathbf{x}}_{k}^{N}$.
%%
%In other words, in order to , the MAP estimate of $\mathbf{x}_{1}^{k-1}$ is plugged into the expression $\sum_{i=1}^{k-1}  b\left(\mathbf{x}_{i}^{N}\right)$.
%%
%%
%%This bias term is in contrast to the Fano metric \cite{fano} which requires the marginalization over the
%%undecided symbols $X_{1}^{k-1}$.
%
Though the genie-aided path metric offers a substantial performance gain,
it is impractical to incorporate it into tree search due to the high complexity associated with the MAP estimation.

%%%%%%%%%%%%%%%%%%%%%%%%%%%%%%%%%%%%%%%%%%%%%%%%%%%%%%%%%%%%%%%%%%%%%%%%%%%%%%%%%%%%%%%%%%%%%%%%%%%
%%%%%%%%%%%%%%%%%%%%%%%%%%%%%%%%%%%%%%%%%%%%%%%%%%%%%%%%%%%%%%%%%%%%%%%%%%%%%%%%%%%%%%%%%%%%%%%%%%%
\subsection{Derivation of Linear Estimate-Based Look-Ahead (LE-LA) Path Metric}
\label{subsec:derivation}
%%%%%%%%%%%%%%%%%%%%%%%%%%%%%%%%%%%%%%%%%%%%%%%%%%%%%%%%%%%%%%%%%%%%%%%%%%%%%%%%%%%%%%%%%%%%%%%%%%%
%%%%%%%%%%%%%%%%%%%%%%%%%%%%%%%%%%%%%%%%%%%%%%%%%%%%%%%%%%%%%%%%%%%%%%%%%%%%%%%%%%%%%%%%%%%%%%%%%%%
In order to alleviate the complexity associated with MAP detection of $\mathbf{x}_{1}^{k-1}$ in the genie-aided path metric, we relax the finite alphabet constraint of  $\mathbf{x}_{1}^{k-1}$ and then replace the MAP estimate by  the linear MMSE estimate $\hat{\mathbf{x}}_{1}^{k-1}$.
Note that when $\hat{\mathbf{x}}_{1}^{k-1}$ is assumed to be Gaussian, the MAP estimate is identical to the linear MMSE
estimate \cite{poor_book}.
%In this subsection, we derive a new path metric based on an linear estimate of undecided symbols under the MMSE criterion. In the derivation, we will
For  a particular path visited $\mathbf{x}_{k}^{N}$, we first define the LE-LA path metric
\begin{definition}
The \emph{linear estimate-based look-ahead path metric}, denoted by $\gamma^{(l)}\left(\mathbf{x}_{k}^{N}\right)$,  is defined as
%%%%%%%%%%%%%%%%%%%%%%%%%%%%%%%%%%%%%%%%%%%%%%%%%%%%%%%%%%%%%%%%%%%%%%%%%%%%%%%%%%%%%%%%%%%%%%%%%%%%%%%%%%%
\begin{equation} \label{eq:npm}
\gamma^{(l)}\left(\mathbf{x}_{k}^{N}\right)  \triangleq \gamma^{(c)}\left(\mathbf{x}_{k}^{N}\right) + \underbrace{\sum_{i=1}^{k-1} b\left(\mathbf{x}_{i}^{N}\right) \bigg|_{\mathbf{x}_{1}^{k-1} = \hat{\mathbf{x}}_{1}^{k-1}} }_{\mbox{bias term, } \gamma^{(b)}\left(\mathbf{x}_{k}^{N}\right) },
%\underbrace{\sum_{i=1}^{k-1} b\left(\begin{bmatrix} \hat{\mathbf{x}}_{1}^{k-1} \\ \mathbf{x}_{k}^{N}\end{bmatrix}\right)}_{\mbox{bias term } s^{(N)}\left(\mathbf{x}_{k}^{N}\right)},
\end{equation}
where $\hat{\mathbf{x}}_{1}^{k-1}$ is the linear MMSE estimate of $\mathbf{x}_{1}^{k-1}$.
% derived assuming that  $\mathbf{x}_{k}^{N} = \widetilde{\mathbf{x}}_{k}^{N}$.
%%%%%%%%%%%%%%%%%%%%%%%%%%%%%%%%%%%%%%%%%%%%%%%%%%%%%%%%%%%%%%%%%%%%%%%%%%%%%%%%%%%%%%%%%%%%%%%%%%%%%%%%%%%
\end{definition}
%p^{(b)}\left(\mathbf{x}_{k}^{N}\right)
%
Note that $\hat{\mathbf{x}}_{1}^{k-1}$ is obtained under the condition that $\mathbf{x}_{k}^{N} = \widetilde{\mathbf{x}}_{k}^{N}$. In the sequel, we denote this bias
 term as $\gamma^{(b)}\left(\mathbf{x}_{k}^{N}\right)$.

To derive the linear MMSE estimate $\hat{\mathbf{x}}_{1}^{k-1}$, we partition the vectors $\mathbf{y}$ and $\mathbf{n}(\triangleq \mathbf{Q}_{1}^{H}\mathbf{n}_{o})$ to $(k-1) \times 1$ and $(N-k+1) \times 1$ vectors, i.e.,
\begin{align}
\mathbf{y} =& \begin{bmatrix} \mathbf{y}_{1}^{k-1} \\ \mathbf{y}_{k}^{N} \end{bmatrix} \label{eq:partition22}
= \begin{bmatrix}\mathbf{R}_{11,k} & \mathbf{R}_{12,k} \\ \mathbf{0} & \mathbf{R}_{22,k} \end{bmatrix}\begin{bmatrix} \mathbf{x}_{1}^{k-1} \\ \mathbf{x}_{k}^{N} \end{bmatrix} + \begin{bmatrix} \mathbf{n}_{1}^{k-1} \\ \mathbf{n}_{k}^{N} \end{bmatrix},
\end{align}
where $\mathbf{R}_{11,k}$, $\mathbf{R}_{12,k}$, and $\mathbf{R}_{22,k}$ are the adequately partitioned sub-matrices of $\mathbf{R}$.
%
%%%%%%%%%%%%%%%%%%%%%%%%%%%%%%%%%%%%%%%%%%%%%%%%%%%%%%%%%%%%%%%%%%%%%%%%%%%%%%%%%%%%%%%%%%%%%%%%%%%%%%%%%%%%%%%%%%%%%%%%%%%%%%
%\begin{align}\label{eq:partition22}
%\begin{bmatrix} \mathbf{y}_{1}^{k-1} \\ \mathbf{y}_{k}^{N} \end{bmatrix} &= \begin{bmatrix}\mathbf{R}_{11,k} & \mathbf{R}_{12,k} \\ \mathbf{0} & \mathbf{R}_{22,k} \end{bmatrix}\begin{bmatrix} \mathbf{x}_{1}^{k-1} \\ \mathbf{x}_{k}^{N} \end{bmatrix} + \begin{bmatrix} \mathbf{n}_{1}^{k-1} \\ \mathbf{n}_{k}^{N} \end{bmatrix}.
%\end{align}
%%%%%%%%%%%%%%%%%%%%%%%%%%%%%%%%%%%%%%%%%%%%%%%%%%%%%%%%%%%%%%%%%%%%%%%%%%%%%%%%%%%%%%%%%%%%%%%%%%%%%%%%%%%%%%%%%%%%%%%%%%%%%%
Using (\ref{eq:partition22}), $\gamma^{(l)}\left(\mathbf{x}_{k}^{N}\right)$ can be expressed as
\begin{align}
\gamma^{(l)}\left(\mathbf{x}_{k}^{N}\right) &= \gamma^{(c)}\left(\mathbf{x}_{k}^{N}\right)+ \gamma^{(b)}\left(\mathbf{x}_{k}^{N}\right)
\end{align}
where
\begin{align}
 \label{eq:biastm1}
\gamma^{(c)}\left(\mathbf{x}_{k}^{N}\right) &= \left\|\mathbf{y}_{k}^{N} - \mathbf{R}_{22,k}\mathbf{x}_{k}^{N} \right\|^{2} + \xi\left(\mathbf{x}_{k}^{N}\right)
\\
 \gamma^{(b)}\left(\mathbf{x}_{k}^{N}\right) &=  \left\|
 \mathbf{y}_{1}^{k-1}  -  \mathbf{R}_{11,k}\hat{\mathbf{x}}_{1}^{k-1}  - \mathbf{R}_{12,k}\mathbf{x}_{k}^{N} \right\|^{2} \label{eq:biastm2}
 %\\ &- \sigma_{n}^{2} \sum_{i=1}^{k-1}   \sum_{j=1}^{Q} \ln  Pr\left(\overline{c}_{i,j}\right).
\end{align}
and $\xi\left(\mathbf{x}_{k}^{N}\right) = - \sigma_{n}^{2} \sum_{i=k}^{N}   \sum_{j=1}^{Q} \ln  Pr\left(\overline{c}_{i,j}\right)$.
%%%%%%%%%%%%%%%%%%%%%%%%%%%%%%%%%%%%%%%%%%%%%%%%%%%%%%%%%%%%%%%%%%%%%%%%%%%%%%%%%%%%%%%%%%%%%%%%%%%%%%%%%%%%%%%%%%%%%%%%%%%%%%
%
Note that the term generated by \emph{a priori} information $\xi\left(\mathbf{x}_{k}^{N}\right)$ considers only $\mathbf{x}_{k}^{N}$  since  the symbols $\mathbf{x}_{1}^{k-1}$ are undecided.
%Since the ISS-MA compares paths of the same length and the second term in the right-hand side in (\ref{eq:biastm2}) is the same for those paths so that it can be ignored.
%
Note also that the linear MMSE estimate of the non-causal symbols $\mathbf{x}_{1}^{k-1}$ is given by \cite{poor_book}
%%%%%%%%%%%%%%%%%%%%%%%%%%%%%%%%%%%%%%%%%%%%%%%%%%%%%%%%%%%%%%%%%%%%%%%%%%%%%%%%%%%%%%%%%%%%%%%%%%%%%%%%%%%%%%%%%%%%%%%%%%%%%
\begin{align}
\hat{\mathbf{x}}_{1}^{k-1} &=
 \mathbf{F}_{k}
 \left(\mathbf{y}_{1}^{k-1} - E\left[\mathbf{y}_{1}^{k-1}\Big|\mathbf{x}_{k}^{N} = \widetilde{\mathbf{x}}_{k}^{N}\right] \right) + E\left[\mathbf{x}_{1}^{k-1}\Big|\mathbf{x}_{k}^{N} = \widetilde{\mathbf{x}}_{k}^{N}\right] \\ \label{eq:softmmse2}
 & = \mathbf{F}_{k}
 \left(\mathbf{y}_{1}^{k-1} -  \mathbf{R}_{11,k}\overline{\mathbf{x}}_{1}^{k-1} - \mathbf{R}_{12,k}\mathbf{x}_{k}^{N} \right) + \overline{\mathbf{x}}_{1}^{k-1},
  \end{align}
where $\overline{\mathbf{x}}_{1}^{k-1} = E[\mathbf{x}_{1}^{k-1}]$ and $\mathbf{F}_{k} = {\rm Cov}(\mathbf{x}_{1}^{k-1},\mathbf{y}_{1}^{k-1}|\mathbf{x}_{k}^{N} =  \widetilde{\mathbf{x}}_{k}^{N})$
$ {\rm Cov}^{-1}(\mathbf{y}_{1}^{k-1}|\mathbf{x}_{k}^{N} = \widetilde{\mathbf{x}}_{k}^{N})$.  We can obtain $\overline{\mathbf{x}}_{1}^{k-1}$ and $\mathbf{F}_{k}$  from \emph{a priori} LLRs as \cite{haykin}
%%%%%%%%%%%%%%%%%%%%%%%%%%%%%%%%%%%%%%%%%%%%%%%%%%%%%%%%%%%%%%%%%%%%%%%%%%%%%%%%%%%%%%%%%%%%%%%%%%%%%%%%%%%%%%%%%%%%%%%%%%%%%
\begin{equation} \label{eq:mean}
\overline{\mathbf{x}}_{1}^{k-1} =   \begin{bmatrix}  \sum_{\theta \in \Theta}\theta \prod_{j=1}^{Q} \frac{1}{2} \left(1 + \overline{c}_{1,j} \tanh\left(\frac{L_{\rm pri}(\overline{c}_{1,j})}{2}\right) \right) \\ \vdots \\  \sum_{\theta \in \Theta}\theta \prod_{j=1}^{Q} \frac{1}{2} \left(1 + \overline{c}_{k-1,j} \tanh\left(\frac{L_{\rm pri}(\overline{c}_{k-1,j})}{2}\right) \right) \end{bmatrix}
\end{equation}
%%%%%%%%%%%%%%%%%%%%%%%%%%%%%%%%%%%%%%%%%%%%%%%%%%%%%%%%%%%%%%%%%%%%%%%%%%%%%%%%%%%%%%%%%%%%%%%%%%%%%%%%%%%%%%%%%%%%%%%%%%%%%
\begin{equation} \label{eq:var}
\mathbf{F}_{k} = \mathbf{\Lambda}_{k}(\mathbf{R}_{11,k})^{H}\left(\left(\mathbf{R}_{11,k}\right)\mathbf{\Lambda}_{k}(\mathbf{R}_{11,k})^{H} + \sigma_{n}^{2}\mathbf{I}\right)^{-1},
\end{equation}
%%%%%%%%%%%%%%%%%%%%%%%%%%%%%%%%%%%%%%%%%%%%%%%%%%%%%%%%%%%%%%%%%%%%%%%%%%%%%%%%%%%%%%%%%%%%%%%%%%%%%%%%%%%%%%%%%%%%%%%%%%%%%
where $\mathbf{\Lambda}_{k} = {\rm diag}\left(\lambda_{1}, \cdots, \lambda_{k-1}\right)$ and $\lambda_{i} = \sum_{\theta \in \Theta}|\theta - \overline{x}_{i}|^{2}\prod_{q=1}^{Q} \frac{1}{2}(1 + \overline{c}_{i,q} \tanh(\frac{L_{\rm pri}(\bar{c}_{i,q})}{2}) )$.
%%%%%%%%%%%%%%%%%%%%%%%%%%%%%%%%%%%%%%%%%%%%%%%%%%%%%%%%%%%%%%%%%%%%%%%%%%%%%%%%%%%%%%%%%%%%%%%%%%%%%%%%%%%%%%%%%%%%%%%%%%%%%
%\begin{equation}
%\mathbf{\Lambda}_{k} = \begin{bmatrix}\lambda_{1} & 0 & 0  \\ 0 & \ddots & 0 \\ 0  & \hdots & \lambda_{k-1}\end{bmatrix},
%\end{equation}
%%%%%%%%%%%%%%%%%%%%%%%%%%%%%%%%%%%%%%%%%%%%%%%%%%%%%%%%%%%%%%%%%%%%%%%%%%%%%%%%%%%%%%%%%%%%%%%%%%%%%%%%%%%%%%%%%%%%%%%%%%%%%
%%%%%%%%%%%%%%%%%%%%%%%%%%%%%%%%%%%%%%%%%%%%%%%%%%%%%%%%%%%%%%%%%%%%%%%%%%%%%%%%%%%%%%%%%%%%%%%%%%%%%%%%%%%%%%%%%%%%%%%%%%%%%
%\begin{equation}
%
%\end{equation}
%%%%%%%%%%%%%%%%%%%%%%%%%%%%%%%%%%%%%%%%%%%%%%%%%%%%%%%%%%%%%%%%%%%%%%%%%%%%%%%%%%%%%%%%%%%%%%%%%%%%%%%%%%%%%%%%%%%%%%%%%%%%%
The set $\Theta$ includes all possible constellation points. In the first iteration of the IDD where \emph{a priori} LLRs, $L_{\rm pri}(\bar{c}_{i,q})$ are unavailable, $\mathbf{\Lambda}_{k} = \mathbf{I}$
and $\overline{\mathbf{x}}_{1}^{k-1} = \mathbf{0}$.
% and they approach $\mathbf{0}$ and $\mathbf{x}_{1}^{k-1}$ as the iteration proceeds.

Using (\ref{eq:biastm2}) and (\ref{eq:softmmse2}),   $\gamma^{(b)}\left(\mathbf{x}_{k}^{N}\right)$ can be rewritten
%
%%%%%%%%%%%%%%%%%%%%%%%%%%%%%%%%%%%%%%%%%%%%%%%%%%%%%%%%%%%%%%%%%%%%%%%%%%%%%%%%%%%%%%%%%%%%%%%%%%%%%%%%%%%%%%%%%%%%%%%%%%%%%
\begin{align}
\gamma^{(b)}\left(\mathbf{x}_{k}^{N}\right) &=\left\| \left(\mathbf{I} - \mathbf{R}_{11,k}\mathbf{F}_{k}\right)\left(
 \mathbf{y}_{1}^{k-1} -  \mathbf{R}_{11,k}\overline{\mathbf{x}}_{1}^{k-1} - \mathbf{R}_{12,k}\mathbf{x}_{k}^{N}\right) \right\|^{2} \\ \label{eq:bb2}  &=   \left\|\mathbf{Z}_{k} \left(
 \mathbf{y}_{1}^{k-1} -  \mathbf{R}_{11,k}\overline{\mathbf{x}}_{1}^{k-1} -  \mathbf{R}_{12,k}\mathbf{x}_{k}^{N}\right) \right\|^{2} ,
\end{align}
%%%%%%%%%%%%%%%%%%%%%%%%%%%%%%%%%%%%%%%%%%%%%%%%%%%%%%%%%%%%%%%%%%%%%%%%%%%%%%%%%%%%%%%%%%%%%%%%%%%%%%%%%%%%%%%%%%%%%%%%%%%%%
where
%%%%%%%%%%%%%%%%%%%%%%%%%%%%%%%%%%%%%%%%%%%%%%%%%%%%%%%%%%%%%%%%%%%%%%%%%%%%%%%%%%%%%%%%%%%%%%%%%%%%%%%%%%%%%%%%%%%%%%%%%%%%%
\begin{align}
\mathbf{Z}_{k} &= \mathbf{I} - \mathbf{R}_{11,k}\mathbf{F}_{k} \\ \label{eq:z2}
&= \sigma_{n}^{2}\left(\mathbf{R}_{11,k}\mathbf{\Lambda}_{k}(\mathbf{R}_{11,k})^{H} + \sigma_{n}^{2}\mathbf{I}\right)^{-1}.
\end{align}
%%%%%%%%%%%%%%%%%%%%%%%%%%%%%%%%%%%%%%%%%%%%%%%%%%%%%%%%%%%%%%%%%%%%%%%%%%%%%%%%%%%%%%%%%%%%%%%%%%%%%%%%%%%%%%%%%%%%%%%%%%%%%
%
%
Further, denoting $\mathbf{q}_{k}=\mathbf{Z}_{k} (
 \mathbf{y}_{1}^{k-1} -  \mathbf{R}_{11,k}\overline{\mathbf{x}}_{1}^{k-1})$ and $\mathbf{P}_{k} = \mathbf{Z}_{k} \mathbf{R}_{12,k}$,  $\gamma^{(l)}\left(\mathbf{x}_{k}^{N}\right)$ can be simply expressed
\begin{equation} \label{eq:finaldef}
\gamma^{(l)}\left(\mathbf{x}_{k}^{N}\right)  = \gamma^{(c)}\left(\mathbf{x}_{k}^{N}\right) +\underbrace{\left\|\mathbf{q}_{k} - \mathbf{P}_{k} \mathbf{x}_{k}^{N} \right\|^{2}}_{\mbox{bias term}}.
\end{equation}
Note that the bias term ${\|\mathbf{q}_{k} - \mathbf{P}_{k} \mathbf{x}_{k}^{N} \|^{2}}$ of the LE-LA path metric can be computed with only linear operations.
Note also that \emph{a priori} information obtained from the channel decoder is reflected through $\overline{\mathbf{x}}_{1}^{k-1}$
 and $\mathbf{\Lambda}_{k}$ in the bias term.

%%%%%%%%%%%%%%%%%%%%%%%%%%%%%%%%%%%%%%%%%%%%%%%%%%%%%%%%%%%%%%%%%%%%%%%%%%%%%%%%%%%%%%%%%%%%%%%%%%%
%%%%%%%%%%%%%%%%%%%%%%%%%%%%%%%%%%%%%%%%%%%%%%%%%%%%%%%%%%%%%%%%%%%%%%%%%%%%%%%%%%%%%%%%%%%%%%%%%%%
\subsection{Efficient Computation of Path Metric}
\label{subsec:efficient}
%%%%%%%%%%%%%%%%%%%%%%%%%%%%%%%%%%%%%%%%%%%%%%%%%%%%%%%%%%%%%%%%%%%%%%%%%%%%%%%%%%%%%%%%%%%%%%%%%%%
%%%%%%%%%%%%%%%%%%%%%%%%%%%%%%%%%%%%%%%%%%%%%%%%%%%%%%%%%%%%%%%%%%%%%%%%%%%%%%%%%%%%%%%%%%%%%%%%%%%
In this subsection, we discuss how the LE-LA path metric can be computed efficiently.
%The bias term in the improved path metric can be obtained at the cost of additional complexity.
%
Recalling that the bias term is expressed as $\|\mathbf{Z}_{k} (
 \mathbf{y}_{1}^{k-1} -  \mathbf{R}_{11,k}\overline{\mathbf{x}}_{1}^{k-1} -  \mathbf{R}_{12,k}\mathbf{x}_{k}^{N}) \|^{2}$, computation of
the path metric is divided into two steps; 1) computation of $\mathbf{Z}_{k}$ for all $k$ prior to the tree search and 2) recursive update
of the path metric for each branch extension during the search.

First, using a matrix inversion formula for block matrices \cite[\emph{Appendix 1.1.3}]{kai}, the operators $\mathbf{Z}_{k} (k = 1, \cdots, N)$ in (\ref{eq:z2}) can be computed recursively.
%In this subsection, we present the procedure to obtain the path metric the ISS-MA.
%
%First, given the channel matrix   $\mathbf{R}$,   $\mathbf{Z}_{1}, \cdots, \mathbf{Z}_{N}$ in (\ref{eq:z2}) are obtained prior to the tree search.
%
% the matrix $\mathbf{Z}_{k+1}$ can be computed from $\mathbf{Z}_{k}$.
%
% For current node $\mathbf{x}_{k}^{N}$, the $k$-by-$k$ matrix $\mathbf{Z}$  is denoted as $\mathbf{Z}^{(k)}$.
%
Denoting
\begin{equation}
\mathbf{R}_{11,k+1} = \begin{bmatrix}\mathbf{R}_{11,k} & {\mathbf{r}}_{k+1}
\\ \mathbf{0} & r_{k+1,k+1}\end{bmatrix}, \; \mathbf{\Lambda}_{k+1} = \begin{bmatrix}\mathbf{\Lambda}_{k} & \mathbf{0}
\\ \mathbf{0} & \lambda_{k+1}\end{bmatrix},
\end{equation}
and  ${\mathbf{r}}_{k+1} = \left[r_{1,k+1}, \cdots, r_{k,k+1} \right]^{T}$, then $\mathbf{Z}_{k+1}$ is expressed as a function of $\mathbf{Z}_{k}$ as
\begin{align} \label{eq:recursive}
\mathbf{Z}_{k+1} = \begin{bmatrix}\mathbf{Z}_{k} - K\lambda_{k+1}\mathbf{Z}_{k}{\mathbf{r}}_{k+1}
{\mathbf{r}}_{k+1}^{H}\mathbf{Z}_{k} & -K\lambda_{k+1}r_{k+1,k+1}\mathbf{Z}_{k}{\mathbf{r}}_{k+1} \\ -K\lambda_{k+1}r_{k+1,k+1}{\mathbf{r}}_{k+1}^{H}\mathbf{Z}_{k} & K\left(\lambda_{k+1}{\mathbf{r}}_{k+1}^{H}\mathbf{Z}_{k}{\mathbf{r}}_{k+1} + \sigma_{n}^{2}\right) \end{bmatrix},
\end{align}
%
%\begin{align} \label{eq:recursive}
%\mathbf{Z}_{k+1} = \begin{bmatrix}\mathbf{Z}_{11} & \mathbf{Z}_{12} \\ \mathbf{Z}_{21} & \mathbf{Z}_{22} \end{bmatrix},
%\end{align}
%where
%\begin{align}
%\mathbf{Z}_{11} =& \mathbf{Z}_{k} - K\lambda_{k+1}\mathbf{Z}_{k}{\mathbf{r}}_{k+1}
%{\mathbf{r}}_{k+1}^{H}\mathbf{Z}_{k} \\
%\mathbf{Z}_{12} =& -K\lambda_{k+1}r_{k+1,k+1}\mathbf{Z}_{k}{\mathbf{r}}_{k+1} \\
%\mathbf{Z}_{21} =& -K\lambda_{k+1}r_{k+1,k+1}{\mathbf{r}}_{k+1}^{H}\mathbf{Z}_{k} \\
%\mathbf{Z}_{22} =& K\left(\lambda_{k+1}{\mathbf{r}}_{k+1}^{H}\mathbf{Z}_{k}{\mathbf{r}}_{k+1} + \sigma_{n}^{2}\right),
%\end{align}
where
\begin{equation}
K = \frac{1}{\lambda_{k+1}\left({\mathbf{r}}_{k+1}^{H}\mathbf{Z}_{k}
{\mathbf{r}}_{k+1} + r_{k+1,k+1}^{2}\right) + \sigma_{n}^{2}}.
\end{equation}
In particular, $\mathbf{Z}_{2} = \frac{\sigma_{n}^{2}}{\lambda_{1}r_{1,1}^{2}+\sigma_{n}^{2}}$.
See Appendix \ref{app:recursive} for the derivation of (\ref{eq:recursive}).
If the \emph{a priori} LLRs are all zero, $\mathbf{Z}_{k}$ does not need to be computed for every symbol as long as the channel remains constant.
If the \emph{a priori} LLRs are non-zero, these steps are performed for each symbol.
However, the required computations can be further reduced by replacing the instantaneous covariance matrix
$\mathbf{\Lambda}_{k}$ by its
time-average over a coherent time \cite{singer}.

%In order to make this computation time-invariant, we may replace the instantaneous
%covariance matrix,
%$\mathbf{\Lambda}_{k}$ by its
%time-average over a coherent time \cite{singer}.
%%

%The path metric
%  at the $k$th level  $p^{(M)}\left(\mathbf{x}_{i}^{N}\right)$  can be obtained from
%  that at the $k+1$th level
%  $p^{(M)}\left(\mathbf{x}_{i+1}^{N}\right)$.
%  If we define an $i\times1$ dummy vector as $\mathbf{d}_{i}$
%  \begin{gather}
%\begin{bmatrix} \mathbf{d}_{i} \\ e_{i} \end{bmatrix} =   \mathbf{d}_{i+1}  - \left[r_{1,i}, \cdots, r_{i,i}\right]^{T} x_{i} \\
%  p^{(T)}\left(\mathbf{x}_{i}^{N}\right) = p^{(T)}\left(\mathbf{x}_{i+1}^{N}\right)  + |e_{i}|^{2} +\xi\left(x_{i}\right) \\
%  p^{(M)}\left(\mathbf{x}_{i}^{N}\right) = p^{(T)}\left(\mathbf{x}_{i}^{N}\right) + \|\mathbf{q}_{k} - \mathbf{P}_{k}\mathbf{x}_{k}^{N}\|^{2}.
% \end{gather}
% where $\mathbf{p}_{i}$ is a dummy vector defined for the $i$th level.

%The computation of the bias term, i.e., $\|\mathbf{q}_{k} - \mathbf{V}_{k}\mathbf{x}_{k}^{N} \|^{2}$ causes the additional complexity over the traditional path metric.
%

%There is also an additional complexity to compute the bias term $\|\mathbf{q}_{k} - \mathbf{P}_{k}\mathbf{x}_{k}^{N} \|^{2}$.
Next, the LE-LA path metric can be recursively updated for each tree extension.
At the root node, a vector is defined such that $\mathbf{a}_{N+1}  = \mathbf{y} - \mathbf{R} \overline{\mathbf{x}}_{1}^{N}$.
The vector $\mathbf{a}_{k}$  is updated from that of its parent node  as
\begin{align}
\begin{bmatrix}
\mathbf{a}_{k}
\\
v_{k}
\end{bmatrix} = \mathbf{a}_{k+1} - \begin{bmatrix} r_{1,k} & \cdots & r_{k,k} \end{bmatrix}^{T} (x_{k} - \overline{x}_{k}),
\end{align}
where $v_{k}$ is a scaler variable.
Using the vector for each path $\mathbf{x}_{k}^{N}$, the LE-LA path metric can be obtained as
\begin{align}
\gamma^{(l)}\left(\mathbf{x}_{k}^{N}\right) &= \gamma^{(c)}\left(\mathbf{x}_{k}^{N}\right) + \gamma^{(b)}\left(\mathbf{x}_{k}^{N}\right)\\
\gamma^{(c)}\left(\mathbf{x}_{k}^{N}\right) &= \gamma^{(c)}\left(\mathbf{x}_{k+1}^{N}\right) + \left|v_{k}\right|^{2} + \xi\left({x}_{k}\right) \\ \label{eq:dum3}
\gamma^{(b)}\left(\mathbf{x}_{k}^{N}\right) &= \left\|\mathbf{Z}_{k} \cdot \mathbf{a}_{k} \right\|^{2},
\end{align}
where $\gamma^{(c)}\left(\mathbf{x}_{N+1}^{N}\right) = 0$.
Noting that the dimension of the matrix $\mathbf{Z}_{k}$ is $(k-1) \times (k-1)$, the number of complex multiplications for the bias term computation is proportional to  $(k-1)^2$.
%
%Hence, the complexity for obtaining the bias term increases with $k$.
%
In order to reduce the complexity, we can look ahead only $N_{l} (< k-1)$ symbols instead of all non-causal symbols.
%
%This implies that only $N_{t}$ symbol estimates are used to derive the bias term.
%
Towards this goal, we set $\alpha = \max(0,k-N_{l})$ and repartition  the system as
\begin{align} \label{eq:newdef}
 \begin{bmatrix} \mathbf{y}_{\alpha}^{k-1} \\ \mathbf{y}_{k}^{N} \end{bmatrix} = \begin{bmatrix}\overline{\mathbf{R}}_{11,k} & \overline{\mathbf{R}}_{12,k} \\ \mathbf{0} & \overline{\mathbf{R}}_{22,k} \end{bmatrix}\begin{bmatrix} \mathbf{x}_{\alpha}^{k-1} \\ \mathbf{x}_{k}^{N} \end{bmatrix} + \begin{bmatrix} \mathbf{n}_{\alpha}^{k-1} \\ \mathbf{n}_{k}^{N} \end{bmatrix},
\end{align}
where $\overline{\mathbf{R}}_{11,k}$ and $\overline{\mathbf{R}}_{12,k}$ are the redefined sub-matrices of (\ref{eq:partition22}), respectively.
In this case, the bias term defined in Section \ref{subsec:derivation} needs to be modified based on this partitioning. In doing so,  the dimension
of $\overline{\mathbf{R}}_{11,k}$ and $\mathbf{Z}_{k}$ is reduced from $(k-1)\times(k-1)$ to $N_l \times N_l$.
%Then, all subsequent computations are carried out based on $\mathbf{y}_{k-N_{l}}^{k-1}$ and $\mathbf{x}_{k-N_{l}}^{k-1}$ instead of $\mathbf{y}_{1}^{k-1}$ and $\mathbf{x}_{1}^{k-1}$.
%Hence, this leads to complexity reduction.
%
The recursive computation of $\mathbf{Z}_{k}$ employing the new partitioning can be derived without matrix inversion (see \cite[{\it Section III. A}]{singer1}).
In addition, in (\ref{eq:dum3}), we only need to multiply $\mathbf{Z}_{k}$ with the last $\alpha$ elements of $\mathbf{a}_{k}$.
%
%The complexity is reduced from $(k-1)^2$ to $N_l^{2}$.
%As shown in Section V, the impact  on detection performance of the bias metric decreases as $k$ gets smaller. Hence, for further reduction in complexity, we can skip computing  the
%bias term for $L_{\rm skip}$ lowermost levels, i.e., $k \leq L_{\rm skip}$.
%
Overall, by using only $N_l$ non-causal symbols for the bias term,  the number of operations for the bias computation can be
reduced from $M\sum_{k=1}^N (k-1)^2(={M}/{6}\cdot (2N^3-3N^2+N))$ to $M N\cdot N_l^2$.

\subsection{Application to APP Detection}
\label{subsec:application}
\begin{table}[t]
\begin{center}
\caption{Summary of ISS-MA}
\begin{tabular}{p{15cm}}
  \hline
  \hline \\[-0.0cm]
  \textbf{Output}: $\left\{L_{\rm post}(\overline{c}_{k,i})\right\}_{k=[1:N], i=[1:Q]}$ \\ \textbf{Input}: $\mathbf{y}$, $\mathbf{H}$, $\left\{L_{\rm pri}(\overline{c}_{k,i})\right\}_{k=[1:N], i=[1:Q]}$, $N_{l}$ and $J$ \\[0.3cm]
  \hline \\[-0.0cm]
  \textbf{STEP~1}: (Preprocessing)  Order $\mathbf{x}$ and $\mathbf{H}$ according to V-BLAST ordering \cite{vblast} or B-Chase preprocessing \cite{milliner}. Then, compute $\mathbf{Z}_{k}$ for all $k$.\\
  \textbf{STEP~2}:  (Initialization) Initialize $i=N+1$ and start the tree search from the root node. \\
  \textbf{STEP~3}: (Loop) Extend $2^Q$ branches for each of  $M$ paths  that have survived at the $(i+1)$th level.
  This generates $2^QM$ paths   at the $i$th level.
   \\
   \textbf{STEP~4}: If $i > 1$, choose the  $M$ best paths with the smallest $\gamma^{(l)}\left(\mathbf{x}_{i}^{N}\right)$
    and go to STEP 3 with $i = i-1$. Otherwise, store all $2^QM$ survival candidates into the list $\mathcal{L}$ and go to STEP 5. \\
    \textbf{STEP~5}: (List extension \& APP calculation) For each value of $k$ and  $i$, compute $\left\{L_{\rm post}(\overline{c}_{k,i})\right\}$ based on $\mathcal{L}$. If the value of $\overline{c}_{k,i}$
   for all elements of $\mathcal{L}$ is either +1 or -1, the value of $\overline{c}_{k,i}$ of the best $J$ candidates (associated with the minimum cost metric)  is flipped and these counter-hypothesis  candidates are added to $\mathcal{L}$
   to generate the extended list $\mathcal{L}_{k}^{\rm ext}$. The APP is calculated over the extended list based on (\ref{eq:mapprom_dasi}). \\[0.3cm]

   \hline
  \hline
  \end{tabular}
 \label{tb:malg}
  \end{center}
\end{table}
In this section, we introduce the soft-input soft-output tree detection algorithm employing the LE-LA path metric.
%
%We focus on  application of the modified path metric to the base-line $M$-algorithm \cite{malg}. However, this path metric can be applicable to other tree search algorithms.
To reduce errors in early detection stages, symbol detection ordering is performed first.
  The V-BLAST ordering \cite{vblast} or B-Chase preprocessing \cite{milliner} can be adopted.
  Note that the B-Chase preprocessing is preferred when $M$ is larger than the constellation size $2^Q$.
%
%\subsection{APP detection}
In each level of the tree, $\gamma^{(l)}\left(\mathbf{x}_{k}^{N}\right)$ of $2^{Q}M$ survival paths are compared and then the $M$ best paths are selected.
Starting from the root node, this candidate selection procedure continues to the bottom level, eventually producing $2^{Q}M$  complete paths.
The symbol vectors corresponding to these complete paths generate a candidate list $\mathcal{L}$, over which the extrinsic LLR for each bit  is calculated.
In the event that a particular bit in each of the candidates takes the same value (all one or zero), the magnitude of the generated LLR might become unduly large, limiting the error-correction capability of the channel decoder \cite{clip}.
In order to prevent this situation, whenever this occurs for the $k$th bit of the candidate list, the $k$th bits of the best $J$  candidates ($J\leq 2^{Q}M$) are flipped and added  into the
candidate list $\mathcal{L}$, generating an extended list $\mathcal{L}_{k}^{\rm ext}$.
As a result, the size of $\mathcal{L}_{k}^{\rm ext}$ becomes $2^{Q}M+J$.
 A more desirable flipping method would be to flip the corresponding bit of all candidates and then select $J$ of them which would have the lowest cost function. Since this method increases the complexity considerably, we employ the alternative that flips the best $J$ candidates. Though the current approach would produce a slightly degraded counter-hypothesis set, we hope that it is highly likely to be, or to at least have a large overlap with, the aforementioned best counter-hypothesis set.
%
%Especially, we have seen that the list extension based on this modification can alleviate performance losses
%with only small values of $J$, (e.g., $J = 4$).
%
Using the list $\mathcal{L}_{k}^{\rm ext}$ together with the max-log approximation, the APP becomes
 \begin{equation} \label{eq:mapprom_dasi}
L_{\rm post}\left(\overline{c}_{k,i}\right) \approx  \max_{\mathbf{x}\in \mathcal{L}_{k}^{\rm ext} \cap X_{k,i}^{+1}}\psi\left(\mathbf{x}\right) - {\max_{\mathbf{x}\in \mathcal{L}_{k}^{\rm ext} \cap X_{k,i}^{-1}}\psi\left(\mathbf{x}\right)}.
\end{equation}
A summary of the ISS-MA is provided in Table \ref{tb:malg}.

\section{Performance Analysis}
\label{sec:performance}

  We discussed in the previous section that the transmitted symbols are always found with $M=1$ if the genie-aided path metric is used.
  Relaxation of the finite alphabet constraint and Gaussian approximation are made for undecided symbols to derive the LE-LA path metric.
  In this section, we analyze the performance of the proposed $M$-algorithm employing the LE-LA path metric.
  %
  %we present the performance analysis of the ISS-MA employing the look-ahead path metric.
%
As a measure for performance, we consider the probability of a CPL event, i.e., the probability that the tree search rejects a path associated with the transmitted symbols.
%Then, we compare it with that of the causal path metric.
%
%To simplify the analysis, we consider the performance for the case of $M=1$ and without ordering of channel columns.
%
In order to make the analysis tractable, we focus on  the case when $M=1$.
Although our analysis focuses only on the case $M=1$, it is clear that lower CPL rate for $M=1$ implies a greater likelihood of a correct symbol being selected for $M>1$ as well.
%since the correct path is derived from the correct path.  %as long as the correct path is included in survived paths.
%
The performance analysis for  $M>1$ is presented via computer simulations in Section \ref{subsec:result}.

Given the channel matrix $\mathbf{R}$ and the \emph{a priori} LLRs,  the probability of CPL can be expressed as
\begin{align}
P_{\rm CPL} &= 1 - \Pr \left(\widetilde{\mathbf{x}} \in \mathcal{L} | \widetilde{\mathbf{x}} \mbox{ is sent}\right) \\
& \label{eq:pcl2} = 1 - \prod_{k=1}^{N} \overline{\Pr} \left( \left. \widetilde{\mathbf{x}}_{k}^{N}  \in \mathcal{L}_{k} \right|\widetilde{\mathbf{x}}_{k+1}^{N} \in \mathcal{L}_{k+1}  \right) \\ \label{eq:pcl3}
& =  1 - \prod_{k=1}^{N} \left(1-\overline{\Pr} \left( \left. \widetilde{\mathbf{x}}_{k}^{N}  \notin \mathcal{L}_{k} \right|\widetilde{\mathbf{x}}_{k+1}^{N} \in \mathcal{L}_{k+1}  \right)\right),
\end{align}
where $\mathcal{L}_{k}$ denotes the  set of the paths selected at the $k$th level and $\overline{\Pr}(\cdot)$ is the probability given that
$\widetilde{\mathbf{x}}$ is sent.
Since we consider the case of $M=1$,  $\widetilde{\mathbf{x}}_{k+1}^{N} \in \mathcal{L}_{k+1}$ implies that a correct path has been selected up to the $(k+1)$th level.
With this setup and from (\ref{eq:partition22}), (\ref{eq:biastm1}), and (\ref{eq:bb2}),
 one can show that  $\gamma^{(l)}\left(\mathbf{x}_{k}^{N}\right)$ is given by
\begin{align}
\gamma^{(l)}\left(\mathbf{x}_{k}^{N}\right) =& \underbrace{\left\|\mathbf{y}_{k}^{N} - \mathbf{R}_{22,k}\begin{bmatrix}x_{k} \\ \widetilde{\mathbf{x}}_{k+1}^{N}\end{bmatrix} \right\|^{2}
+ \xi\left(x_{k}\right) + \xi\left(\widetilde{\mathbf{x}}_{k+1}^{N}\right)}_{\gamma^{(c)}\left(\mathbf{x}_{k}^{N}\right)}  + \underbrace{\left\|
 \mathbf{q}_{k} - \mathbf{P}_{k}\begin{bmatrix}x_{k} \\ \widetilde{\mathbf{x}}_{k+1}^{N}\end{bmatrix} \right\|^{2}}_{\gamma^{(b)}\left(\mathbf{x}_{k}^{N}\right)} \label{eq:ptt2}\\ \nonumber
   =&  \left|r_{k,k}\left(\widetilde{x}_{k} -x_{k}\right) + n_{k}  \right|^{2} + \sum_{i=k+1}^{N}|n_{i}|^{2}  + \left\|\mathbf{Z}_{k} \mathbf{r}_{k}\left(\widetilde{x}_{k}-x_{k}\right) + \mathbf{Z}_{k}\mathbf{b}_{k} \right\|^{2}
   + \xi\left(x_{k}\right) + \xi\left(\widetilde{\mathbf{x}}_{k+1}^{N}\right)  \\ =& \left\|\begin{bmatrix}\mathbf{Z}_{k}\mathbf{r}_{k} \\ r_{k,k}  \end{bmatrix}\left(\widetilde{x}_{k}-x_{k}\right) +
\begin{bmatrix} \mathbf{Z}_{k}\mathbf{b}_{k} \\ n_{k}  \end{bmatrix} \right\|^{2} + \xi (x_{k})+ \sum_{i=k+1}^{N}|n_{i}|^{2} + \xi\left(\widetilde{\mathbf{x}}_{k+1}^{N}\right) \\ \label{eq:ptt4}
=& \left\|\sqrt{\mathbf{r}_{k}^{H}\mathbf{Z}_{k}\mathbf{r}_{k} + |r_{k,k}|^{2}}\left(\widetilde{x}_{k}-x_{k}\right) +
\frac{\begin{bmatrix}\mathbf{Z}_{k}\mathbf{r}_{k} & r_{k,k}  \end{bmatrix}}{\sqrt{\mathbf{r}_{k}^{H}\mathbf{Z}_{k}\mathbf{r}_{k} + |r_{k,k}|^{2}}} \begin{bmatrix} \mathbf{Z}_{k}\mathbf{b}_{k} \\ n_{k}  \end{bmatrix}  \right\|^{2} + \xi (x_{k})+ C
\end{align}
where
 $\mathbf{b}_{k} = \mathbf{R}_{11,k}\left(\widetilde{\mathbf{x}}_{1}^{k-1} - \overline{\mathbf{x}}_{1}^{k-1}\right) + \mathbf{n}_{1}^{k-1}$, and $\mathbf{r}_{k} = \mathbf{R}_{12,k} \mathbf{e}_{1} = \left[r_{1,k}, \cdots, r_{k-1,k}\right]^{T}$.
Note that $C$ is independent of the selection of $x_{k}$.
%
%
%
%
%The sufficient statistics for the detection of $x_{k}$ is given by
%
The first term in (\ref{eq:ptt4}) can be interpreted as the distance metric between the output of a scalar additive noise channel
\begin{equation} \label{eq:additive}
\xi_{k} = \sqrt{\mathbf{r}_{k}^{H}\mathbf{Z}_{k}\mathbf{r}_{k} + |r_{k,k}|^{2}}\widetilde{x}_{k} +
\frac{\left[\mathbf{Z}_{k}\mathbf{r}_{k} \;\; r_{k,k}\right]}{\sqrt{\mathbf{r}_{k}^{H}\mathbf{Z}_{k}\mathbf{r}_{k} + |r_{k,k}|^{2}}} \left[ \left(\mathbf{Z}_{k}\mathbf{b}_{k}\right)^{T} \;\; n_{k}\right]^{T}
 \end{equation}
 and a symbol candidate $\sqrt{\mathbf{r}_{k}^{H}\mathbf{Z}_{k}\mathbf{r}_{k} + |r_{k,k}|^{2}}x_{k}$.
%If the noise has unimodal distribution, the distance metric becomes the cost function for ML detection.
%
The ISS-MA chooses the $M$ best symbols $x_{k}$ according to the cost metric in (\ref{eq:ptt4}).
Since the \emph{a priori} term $\xi (x_{k})$ in (\ref{eq:ptt4}) leads  to better detection, we ignore the impact of it in our discussion.
If we let $E\left[\mathbf{b}_{k} \mathbf{b}_{k}^{H}\right] = \mathbf{\Sigma}_{k}= \left(\mathbf{R}_{11,k}\mathbf{\Lambda}_{k} \mathbf{R}_{11,k}^{H} + \sigma_{n}^{2} \mathbf{I}\right)$ and $\mathbf{Z}_{k} = \sigma_{n}^{2}\mathbf{\Sigma}_{k}^{-1}$, then the  signal to interference plus noise ratio (SINR) of  the scalar additive noise channel is given by
\begin{align} \label{eq:rhos}
{\rm SINR} =&  \frac{\left(\mathbf{r}_{k}^{H}\mathbf{Z}_{k}\mathbf{r}_{k} + |r_{k,k}|^{2}\right)^{2}}{\mathbf{r}_{k}^{H}\mathbf{Z}_{k}^{2}E\left[\mathbf{b}_{k} \mathbf{b}_{k}^{H}\right] \mathbf{Z}_{k}^{2} \mathbf{r}_{k} + \sigma_{n}^{2}|r_{k,k}|^{2}}  \\ \label{eq:rhos2}
=& \frac{1}{\sigma_{n}^{2}}\frac{\left(\mathbf{r}_{k}^{H}\left(\sigma_{n}^{4} \mathbf{\Sigma}_{k}^{-2}\right)\mathbf{r}_{k} + |r_{k,k}|^{2}\right)^{2}}{\mathbf{r}_{k}^{H}\left(\sigma_{n}^{6} \mathbf{\Sigma}_{k}^{-3}\right)\mathbf{r}_{k} + |r_{k,k}|^{2}}.
\end{align}
\begin{lemma}
\label{lemma:sinr_bound}
The SINR in (\ref{eq:rhos2}) is bounded by
\begin{align} \label{eq:bound}
\sigma_{n}^{2} \mathbf{r}_{k}^{H}\mathbf{\Sigma}_{k}^{-2} \mathbf{r}_{k} + \frac{|r_{k,k}|^{2}}{\sigma_{n}^{2}} \leq {\rm SINR} \leq \mathbf{r}_{k}^{H}\mathbf{\Sigma}_{k}^{-1} \mathbf{r}_{k} + \frac{|r_{k,k}|^{2}}{\sigma_{n}^{2}}.
\end{align}
\end{lemma}

\begin{proof}
See Appendix \ref{app:sinr_bound}.
\end{proof}
Taking similar steps, one can show that the SINR  for the causal path metric $\gamma^{(c)}\left(\mathbf{x}_{k}^{N}\right)$  is $\frac{|r_{k,k}|^{2}}{\sigma_{n}^{2}}$.
Hence, $\mathbf{r}_{k}^{H}\mathbf{\Sigma}_{k}^{-1} \mathbf{r}_{k}^{H}$ and $\sigma_{n}^{2} \mathbf{r}_{k}^{H}\mathbf{\Sigma}_{k}^{-2} \mathbf{r}_{k}^{H}$ can be regarded as
upper and lower bounds on the SINR gain achieved by the LE-LA path metric, respectively.
It is of interest to check the behavior of the upper and lower bound of SINR gain for high dimensional systems.
Suppose that   $N,L \rightarrow \infty$ with a fixed aspect ratio $\beta = {N}/{L}$ ($0<\beta \leq1$),
and let $\lambda_{\rm min}$ and $\lambda_{\rm max}$ be smallest and largest diagonals of  $\mathbf{\Lambda}_{k}$, respectively. Then, we attain a looser bound on the SINR is
\begin{align} \label{eq:bound2}
\underbrace{\sigma_{n}^{2} \mathbf{r}_{k}^{H}(\sigma_{n}^{2}\mathbf{I} + \lambda_{\rm max} \mathbf{R}_{11,k}\mathbf{R}_{11,k}^{H})^{-2} \mathbf{r}_{k}}_{B_k^{\rm lower}} + \frac{|r_{k,k}|^{2}}{\sigma_{n}^{2}} \leq {\rm SINR} \leq \underbrace{\mathbf{r}_{k}^{H}(\sigma_{n}^{2}\mathbf{I} + \lambda_{\rm min} \mathbf{R}_{11,k}\mathbf{R}_{11,k}^{H})^{-1} \mathbf{r}_{k}}_{B_k^{\rm upper}} + \frac{|r_{k,k}|^{2}}{\sigma_{n}^{2}},
\end{align}
where the upper and lower bound of the SINR gain are denoted as $B_k^{\rm upper}$ and $B_k^{\rm lower}$, respectively.
 Note that (\ref{eq:bound2}) can be shown by the relationship $\mathbf{B} \preceq \mathbf{\Sigma}_{k} \preceq \mathbf{A} $ (equivalently, $ \mathbf{\Sigma}_{k}^{-1}\preceq \mathbf{B}^{-1} $ and $\mathbf{A}^{-2} \preceq \mathbf{\Sigma}_{k}^{-2}  $), where $\mathbf{A} =  \sigma_{n}^{2}\mathbf{I} + \lambda_{\rm max} \mathbf{R}_{11,k}\mathbf{R}_{11,k}^{H}$ and $\mathbf{B} = \sigma_{n}^{2}\mathbf{I} + \lambda_{\rm min} \mathbf{R}_{11,k}\mathbf{R}_{11,k}^{H}$ and $X \succeq 0$ implies that the matrix $X$ is positive semi-definite.
%
%
%$\mathbf{r}_{k}^{H}\mathbf{\Sigma}_{k}^{-1} \mathbf{r}_{k} \leq \mathbf{r}_{k}^{H}(\sigma_{n}^{2}\mathbf{I} + \lambda_{\rm min} \mathbf{R}_{11,k}\mathbf{R}_{11,k}^{H})^{-1}\mathbf{r}_{k}$. $\mathbf{r}_{k}^{H}\mathbf{\Sigma}_{k}^{-1} \mathbf{r}_{k}  \leq  \mathbf{r}_{k}^{H}(\sigma_{n}^{2}\mathbf{I} + \lambda_{\rm min} \mathbf{R}_{11,k}\mathbf{R}_{11,k}^{H})^{-1}\mathbf{r}_{k}$
%%
%
%%
\begin{theorem}
\label{theorem:random}
 For an $L\times N$  matrix $\mathbf{H}$ whose elements are i.i.d. random variables with zero mean and variance $\frac{1}{L}$, the upper and lower bound of the SINR gain
  for the level
$k = \gamma N + 1 \;\; (0<\gamma<1)$ converge to
\begin{align}
   B_k^{\rm upper} \label{eq:conv_sig} &\longrightarrow  B_k^{{\rm upper},\infty} =  \frac{1}{2\lambda_{\rm min}} \left(-1 - (1 - \gamma\beta)\frac{\lambda_{\rm min}}{\sigma_{n}^{2}}  + G\left(\frac{\lambda_{\rm min}}{\sigma_{n}^{2}}, \gamma\beta \right)
 %\sqrt{1+ \frac{2\left(1+\gamma\beta\right)\lambda_{\rm min}}{\sigma_{n}^{2}} + \frac{(1- \gamma \beta)^{2}\lambda_{\rm min}^{2}}{\sigma_{n}^{4}}}
 \right)
\\ \label{eq:conv_sig2}
B_k^{\rm lower}  & \longrightarrow B_k^{{\rm lower},\infty} = \frac{1}{2\sigma_{n}^{2}} \left(- \left(1 - \gamma\beta\right) + \frac{1+\gamma \beta + (1 - \gamma \beta)^{2} \frac{\lambda_{\rm max}}{\sigma_{n}^{2}}}{G\left(\frac{\lambda_{\rm max}}{\sigma_{n}^{2}}, \gamma\beta \right)}\right)
%\label{eq:conv_sig2} \sigma_{n}^{2} \mathbf{r}_{k}^{H}\mathbf{\Sigma}_{k}^{-2} \mathbf{r}_{k}^{H} &\longrightarrow 0
\end{align}
as $N,L \rightarrow \infty$ with  $\beta = {N}/{L}$, where $G(x,b) = \sqrt{1+2(1+b)x + (1-b)^2x^{2}}$.

\end{theorem}

\begin{proof}
See Appendix \ref{app:random}.
\end{proof}

\begin{corollary}
\label{corol:asymp}
As  $\sigma_{n}^{2} \rightarrow 0$, we have
\begin{align} \label{eq:stop}
B_k^{{\rm upper},\infty} & \longrightarrow {\lambda_{\rm min}}\frac{\gamma\beta}{(1 - \gamma\beta)} \\ \label{eq:stop2}
B_k^{{\rm lower},\infty} & \longrightarrow 0.
\end{align}
In addition, as $\sigma_{n}^{2} \rightarrow 0$, $B_k^{{\rm upper},\infty}$  monotonically increases and approaches ${\lambda_{\rm min}}\frac{\gamma\beta}{(1 - \gamma\beta)} $.
\end{corollary}

%\begin{proof}
%The proof is omitted since it is trivial.
%%See Appendix \ref{app:asymp}.
%\end{proof}

We can deduce from (\ref{eq:stop}) and (\ref{eq:stop2}) that the actual SINR gain approaches a deterministic value  between $[0,{\lambda_{\rm min}}\frac{\gamma\beta}{(1 - \gamma\beta)}]$.
%
%decreases towards  the quantity in (\ref{eq:stop}), which is
%strictly less than one. This is consistent with our observation made in  (\ref{eq:scale}).
%
One can also show that both $B_k^{{\rm upper},\infty}$ and $B_k^{{\rm lower},\infty}$ are an increasing functions of $\gamma\beta \in (0,1)$.
Noting that  $\gamma$ indicates an index for tree depth,
the SINR bounds achieve their maximum at the top level of the tree $(k=N)$.

Next, we analyze the probability of CPL using the SINR obtained.
It is worth taking a close look at the term $\mathbf{Z}_{k}\mathbf{b}_{k}$ in (\ref{eq:additive}).
Recalling that $\mathbf{b}_{k} = \mathbf{R}_{11,k}\left(\widetilde{\mathbf{x}}_{1}^{k-1} - \overline{\mathbf{x}}_{1}^{k-1}\right) + \mathbf{n}_{1}^{k-1}$ and $\mathbf{Z}_{k} = \sigma_{n}^{2}\left(\mathbf{R}_{11,k}\mathbf{\Lambda}_{k}(\mathbf{R}_{11,k})^{H} + \sigma_{n}^{2}\mathbf{I}\right)^{-1}$, $\mathbf{Z}_{k}\mathbf{b}_{k}$ is an MMSE estimate of $\mathbf{n}_{1}^{k-1}$ \cite{poor_book}.
In order to make the derivation of the CPL probability more tractable, we use a Gaussian approximation for the MMSE estimation error $(\mathbf{n}_{k} - \mathbf{Z}_{k}\mathbf{b}_{k})$ or
 equivalently, $\mathbf{Z}_{k}\mathbf{b}_{k}$.
Under this approximation, we can assume that the interference plus noise of the scalar channel is Gaussian.
The validity of this approximation has been supported in many asymptotic scenarios in \cite{poor2} and \cite{ping}.
%Though more accurate analysis based on higher order statistics was presented in \cite{ping},
%our analysis here relies on Gaussian approximation.
%
%
In particular, it is shown  that  the Gaussian approximation is highly accurate for large problem size $N$ \cite{dguo}.
Using the SINR in (\ref{eq:rhos2}), the probability of CPL for the $k$th level detection  can be expressed as \cite{proakis}
\begin{align} \label{eq:exact}
\underline{\Pr}  \left( \left. \widetilde{\mathbf{x}}_{k}^{N}  \notin \mathcal{L}_{k} \right|\widetilde{\mathbf{x}}_{k+1}^{N} \in \mathcal{L}_{k+1}, \mathbf{H}  \right)
\leq  4\left(1-\frac{1}{\sqrt{2^Q}}\right) Q\left(\sqrt{
K\frac{1}{\sigma_{n}^{2}}\frac{\left(\mathbf{r}_{k}^{H}\left(\sigma_{n}^{4} \mathbf{\Sigma}_{k}^{-2}\right)\mathbf{r}_{k} + |r_{k,k}|^{2}\right)^{2}}{\mathbf{r}_{k}^{H}\left(\sigma_{n}^{6} \mathbf{\Sigma}_{k}^{-3}\right)\mathbf{r}_{k} + |r_{k,k}|^{2}}}\right),
\end{align}
where  $K= \frac{3}{\left(2^{Q}-1\right)}$. The inequality in (\ref{eq:exact}) follows from the existence of \emph{a priori} terms in (\ref{eq:ptt4}), which lowers the actual CPL
probability. From (\ref{eq:bound}), we have
\begin{align} \label{eq:cplq}
 \underline{\Pr}  \left( \left. \widetilde{\mathbf{x}}_{k}^{N}  \notin \mathcal{L}_{k} \right|\widetilde{\mathbf{x}}_{k+1}^{N} \in \mathcal{L}_{k+1}, \mathbf{H}  \right)
       \leq 4\left(1-\frac{1}{\sqrt{2^Q}}\right) Q\left(\sqrt{K\left(\sigma_{n}^{2} \mathbf{r}_{k}^{H}\mathbf{\Sigma}_{k}^{-2} \mathbf{r}_{k} + \frac{|r_{k,k}|^{2}}{\sigma_{n}^{2}}
      \right)}\right).
\end{align}
%where $K= \frac{3}{\left(2^{Q}-1\right)}$. %The inequality holds since the term $\xi (x_{k})$ would lead to better detection.
%

%
Using (\ref{eq:cplq}), we can analyze an average probability of CPL for a random channel $\mathbf{H}$ whose elements are independent complex Gaussian with $\mathcal{CN}(0,1)$.
The average probability of CPL, denoted as $\overline{P}_{\rm CPL}$, is given by
\begin{align}
\overline{P}_{\rm CPL} =&   1 - E_{\mathbf{H}}\left[\prod_{k=1}^{N} \left(1-\underline{\Pr} \left( \left. \widetilde{\mathbf{x}}_{k}^{N}  \notin \mathcal{L}_{k} \right|\widetilde{\mathbf{x}}_{k+1}^{N} \in \mathcal{L}_{k+1}, \mathbf{H}  \right)\right)\right] \\
 =& \sum_{k=1}^{N} E_{\mathbf{H}}\left[ \underline{\Pr} \left( \left. \widetilde{\mathbf{x}}_{k}^{N}  \notin \mathcal{L}_{k} \right|\widetilde{\mathbf{x}}_{k+1}^{N} \in \mathcal{L}_{k+1}, \mathbf{H}  \right) \right] + \mbox{higher order terms},
\end{align}
where $E_{\mathbf{H}}[\cdot]$ denotes the expectation over $\mathbf{H}$.
The average CPL probability is obtained after evaluating $E_{\mathbf{H}}\left[ \underline{\Pr} \left( \left. \widetilde{\mathbf{x}}_{k}^{N}  \notin \mathcal{L}_{k} \right|\widetilde{\mathbf{x}}_{k+1}^{N} \in \mathcal{L}_{k+1}, \mathbf{H}  \right) \right]$ for all $k$.
In our analysis, we do not put our emphasis on the higher order terms since they become negligible in the high SNR regime.
Using the relationship $Q(\sqrt{x+y}) \leq Q(\sqrt{x})\exp\left(-\frac{y}{2}\right)$ for $x,y > 0$ and from (\ref{eq:cplq}), we have
\begin{align}
& E_{\mathbf{H}}\left[ \underline{\Pr} \left( \left. \widetilde{\mathbf{x}}_{k}^{N}  \notin \mathcal{L}_{k} \right|\widetilde{\mathbf{x}}_{k+1}^{N} \in \mathcal{L}_{k+1}, \mathbf{H}  \right) \right] \\
  &\leq 4\left(1-\frac{1}{\sqrt{2^Q}}\right) E_{\mathbf{H}}\left[Q\left(\sqrt{K\frac{|r_{k,k}|^{2}}{\sigma_{n}^{2}}}\right) \exp\left(-\frac{K\sigma_{n}^{2}\mathbf{r}_{k}^{H}\mathbf{\Sigma}_{k}^{-2}\mathbf{r}_{k}}{2}
   \right)\right]  \\& =  4\left(1-\frac{1}{\sqrt{2^Q}}\right) E_{\mathbf{H}}\left[Q\left(\sqrt{K\frac{|r_{k,k}|^{2}}{\sigma_{n}^{2}}}\right)\right]\label{eq:avg} E_{\mathbf{H}}\left[ \exp\left(-\frac{K\sigma_{n}^{2}\mathbf{r}_{k}^{H}\mathbf{\Sigma}_{k}^{-2}\mathbf{r}_{k}}{2}
   \right)
   \right] ,
\end{align}
where (\ref{eq:avg}) follows from independence of $r_{k,k}$ and $\mathbf{r}_{k}$.
Noting that $r_{k,k}$ has a Chi-square distribution with $2(L-k+1)$ degrees of freedom and $\mathbf{r}_{k}$ has independent complex Gaussian elements \cite[\emph{Lemma 2.1}]{verdu_mono},
we have \cite{viswanath},
%
%Coincidently, the quantity $\mathbf{r}_{k}^{H}\mathbf{\Sigma}^{-1}\mathbf{r}_{k}$  is identical to the SINR term $T_{k}$ of MMSE detector, which is asymptotically normal and well approximated by
%gamma distribution for the system of finite dimension.
%%
%Accurate approximation of the distribution is out of the scope of this paper.
%
\begin{align}
E_{\mathbf{H}}\left[Q\left(\sqrt{\frac{K}{\sigma_{n}^{2}}|r_{k,k}|^{2}}\right)\right] =& \left(\frac{1}{2} - \frac{1}{2} \sqrt{\frac{K}{K + 2\sigma_{n}^{2}}}\right)^{L-k+1}  \sum_{l=0}^{L-k}{L-k+l \choose l}\left(\frac{1}{2} + \frac{1}{2} \sqrt{\frac{K}{K + 2\sigma_{n}^{2}}}\right)^{l} \label{eq:first1}
%\\\label{eq:first2}
%\approx & {2(L-k)+1 \choose L-k+1}\left(\frac{\sigma_{n}^{2}}{2K}\right)^{L-k+1}.
\end{align}
%
% Note that the derivation of this term does not rely on Gaussian approximation of the MMSE estimation error.
%
%
\begin{lemma}
\label{lemma:scaling}
An upper bound on the scaling gain in (\ref{eq:avg}) is given by
\begin{align}
E_{\mathbf{H}}\left[ \exp\left(-\frac{K\sigma_{n}^{2}\mathbf{r}_{k}^{H}\mathbf{\Sigma}_{k}^{-2}\mathbf{r}_{k}}{2}
   \right)\right] \leq& \int_{0}^{\infty} \cdots  \int_{0}^{\infty} \left(\prod_{i=1}^{k-1}  \frac{1}{1+\frac{K}{2}\frac{\sigma_{n}^{2}}{\left(\lambda_{\rm max}x_{i}+\sigma_{n}^{2}\right)^{2}}}\right) \nonumber \\  \label{eq:scale}
   &\times f_{\eta_{1}, \cdots, \eta_{k-1}}\left(x_{1}, \cdots, x_{k-1}\right) dx_{1} \cdots dx_{k-1},
\end{align}
\end{lemma}
where %$f_{\eta_{1},  \cdots, \eta_{k-1}}\left(x_{1}, \cdots, x_{k-1}\right)$ is the joint distribution of the eigenvalues $\eta_{1}, \cdots, \eta_{k-1}$, i.e.,
\begin{align}\nonumber
f_{\eta_{1}, \cdots, \eta_{k-1}}\left(x_{1}, \cdots, x_{k-1}\right) = \frac{1}{(k-1)!}\exp\left(-\sum_{i=1}^{k-1} x_{i}\right) \prod_{i=1}^{k-1}\frac{x_{i}^{L-k+1}}{\left(k-1-i\right)!(L-i)!} \prod_{i<j}^{k-1} \left(x_{i} - x_{j}\right)^{2}.
\end{align}
\begin{proof}
See Appendix \ref{app:scaling}.
\end{proof}
While  $\exp\left(-\frac{K\sigma_{n}^{2}\mathbf{r}_{k}^{H}\mathbf{\Sigma}_{k}^{-2}\mathbf{r}_{k}}{2}\right)$ in (\ref{eq:scale}) tends to one as $\sigma_{n}^{2} \rightarrow 0$,
 (\ref{eq:first1}) decreases to zero with a slope  $\lim_{\sigma_{n}^{2} \rightarrow 0 }\ln(P_{e})/\ln(\sigma_{n}^{2}) = L-k+1$.
Therefore, at high SNR, the  probability of CPL for the top level $(k=N)$ would dominate, i.e.,
\begin{equation} \label{eq:fin}
\overline{P}_{\rm CPL} \lessapprox 4\left(1-\frac{1}{\sqrt{2^Q}}\right) E_{\mathbf{H}}\left[Q\left(\sqrt{K\frac{|r_{N,N}|^{2}}{\sigma_{n}^{2}}}\right)\right] E_{\mathbf{H}}\left[ \exp\left(-\frac{K\sigma_{n}^{2}\mathbf{r}_{N}^{H}\mathbf{\Sigma}_{N}^{-2}\mathbf{r}_{N}}{2}
   \right)
   \right],
\end{equation}
where the right-hand side is obtained from (\ref{eq:first1}) and (\ref{eq:scale}).
Following similar steps, we can also show that the upper bound of the average CPL probability for the causal path metric  becomes
\begin{equation} \label{eq:fin_causal}
\overline{P}_{\rm CPL}^{\rm causal} \lessapprox  4\left(1-\frac{1}{\sqrt{2^Q}}\right)E_{\mathbf{H}} \left[Q\left(\sqrt{K\frac{|r_{N,N}|^{2}}{\sigma_{n}^{2}}}\right)\right].
 \end{equation}
We observe from (\ref{eq:fin}) that the average CPL probability of the LE-LA path metric is smaller than that of the causal path metric by  the factor of $E_{\mathbf{H}}\left[ \exp\left(-\frac{K}{2}\sigma_{n}^{2}\mathbf{r}_{N}^{H}\mathbf{\Sigma}_{N}^{-2}\mathbf{r}_{N}
   \right)\right]$. Since this term is  strictly less than unity, it corresponds to the scaling gain obtained from the LE-LA path metric.

 In Fig.~\ref{fig:anal}, we provide the plot of the average CPL probability versus SNR  for several system sizes ($N=5,10,15$, and $20$).
We assume uncoded QPSK transmission.
 The average CPL rate and its upper bound are obtained from (\ref{eq:exact}) and (\ref{eq:cplq}).
 %
 %On the other hand, the simulation results are generated from full-Monte-Carlo simulations.
%
For a comprehensive view, we also include the average CPL rate for the causal path metric in (\ref{eq:fin_causal}).
%
%We compared these CPL rates with the simulation results.
%
For all cases considered, the CPL expression in (\ref{eq:exact}) is quite close to that obtained from the simulation results, supporting the accuracy of the analytic bound we obtained.
In particular, the upper bound of the average CPL rate appears tight at high SNR.
Fig.~\ref{fig:scale} shows how the scaling gain in (\ref{eq:scale}) varies as a function of SNR and system size.
We observe that the performance gain of the LE-LA path metric improves with system size and the maximum is achieved in low to moderate SNR range (10 dB $\sim$ 20 dB).
Notice that this behavior is desirable for IDD, since the performance in low-to-mid SNR range is critical in triggering performance improvement though iterations \cite{brink}.
%

%the performance gain of the modified path metric over the traditional one

%Notice that this quantity is strictly less than one, which implies better performance bound compared to that of traditional scheme.
%%
%For sufficiently high SNR,  i.e., $\frac{d_{\rm min}^{2}}{4\sigma_{n}^{2}} \gg 1$,
%%
%\begin{align}
%& E_{\mathbf{H}}\left[ \exp\left(-\frac{d_{\rm dim}^{2}}{4\sigma_{n}^{2}}\mathbf{r}_{k}^{H}\left(\sigma_{n}^{4}\mathbf{\Sigma}^{-2}\right)\mathbf{r}_{k}
%   \right)\right] \nonumber \\
%   &\leq \frac{E_{\mathbf{R}_{11,k}}\left[\prod_{i=1}^{k-1} \left(\eta_{i}+\sigma_{n}^{2}\right)^{2}\right]}{d_{\rm min}^{2(k-1)} \sigma_{n}^{2(k-1)}} \\ &=
%   \frac{E_{\mathbf{R}_{11,k}}\left[{\rm det}^{2}\left(\mathbf{\Sigma}\right)\right]}{d_{\rm min}^{2(k-1)} \sigma_{n}^{2(k-1)}}
%\end{align}

\section{Simulation and Discussion}
\label{sec:simulation}
In this section,  we evaluate the performance of the ISS-MA through computer simulations.
  First, we observe the performance of the soft-input soft-output $M$-algorithm employing the LE-LA path metric and that employing the conventional path metric.
%For comparison, we evaluate the performace  both path metrics to the base-line $M$-algorithm introduced in \cite{jong}.
%
Note that the LE-LA path metric is not restricted to a particular search scheme and can be extended to more sophisticated breadth-first search algorithms (such as \cite{milliner} and \cite{guo}).
Next, we compare the performance-complexity trade-off of the ISS-MA with the existing soft-input soft-output detectors.

%%In particular, we focus performance evaluation for the systems of large size, e.g.,
%% $N = 8$ to   $14$ with 16-QAM modulation.
%Fig.~\ref{fig:lskip} shows the impact of the value of $L_{\rm skip}$ on the bit error rate (BER) for the $10 \times 10$ 16-QAM  system.
%%
%The parameter $M$ is set to $16$ and
%%
%the SNR is chosen to $12.5$ dB at which BER begins to drop sharply towards zero.
%%
%As increasing $L_{\rm skip}$, the BER remains zero but begins to rise at $L_{\rm skip} = 7$.
%%
%Hence, we can conclude that only first three levels are affected by the bias term
%significantly.
%%
%According to the intensive simulations, the choice of $L_{\rm skip} = 2N/3$ leads to negligible performance loss.

\subsection{Simulation Setup}
\label{subsec:setup}
The simulation setup for the IDD system is as follows. A total of $2\times 10^5$ information bits are randomly generated. A rate $R = 1/2$ recursive systematic convolutional (RSC) code with feedback polynomial $1+D+D^2$ and feedforward polynomial $1+D^2$.
We use a random interleaver of size of $12,000$ bits. We use a gray mapping for QAM modulation.
We assume fast fading channels where each entry of $\mathbf{H}$ is i.i.d. complex Gaussian $\mathcal{CN}(0,1)$ and
perfect knowledge of the channel state at the receiver is assumed.
For the  channel decoding, a \emph{max-log-MAP} decoder  \cite{logmapsova} is employed.
% and
%   a total of 10 inner iterations are carried out for each code block.\footnote{For convenience, we call ``inner iteration" for decoding of turbo code  and ``iteration" for iterative detection and decoding. }
%
%
%
%
%
%The size of the frame for turbo iteration is set to 800 symbols, equivalently, $800\times QN_l$ bits.
%%
%A total of 12  iterations are performed.
%%
%and the decoded results are compared with these transmitted bits to measure bit error rate (BER).
%We focus on a $20 \times 20$ QPSK system  and a $10 \times 10$ 16-QAM system, both of which would require quite demanding detection complexity.
%The channel state  is assumed to be constant over a block consisting of ten symbols but changes independently across each individual  block.
%
%
%
%
The SNR is defined as ${\rm SNR} = 10 \log_{10} ({N}/{\sigma_n^2})$.
%
%The spectral efficiency of the MIMO transmission is given by $R Q N$ bit/Hz/s and
%
%
 Computational  complexity of detectors is measured by counting the average number of complex multiplications per symbol period and per  iteration.\footnote{The complexity for QR decomposition and detection ordering is not considered since they are common in all detection algorithms under consideration.}

 \subsection{Simulation Results}
 \label{subsec:result}
First, we compare the performance of the causal path metric and
the LE-LA path metric.
%
%The soft-input soft-output $M$-algorithm employing the causal path metric was presented in .
%
We consider the $12\times 12$ 16-QAM MIMO system, which requires high detection complexity.
For fair comparison, we employ the same candidate extension strategy with $J=16$ (described in Section \ref{subsec:application}) for
both algorithms.
The parameter $N_{l}$ is set to 5 for the LE-LA path metric.
In Fig~\ref{fig:bervsm},  the plots of bit error rate (BER) versus SNR are provided for several $M$ values ($M=4,6,8$ and $12)$.
Each plot shows the BER curves obtained after a different number of iterations.
 The ISS-MA outperforms the conventional $M$-algorithm for all $M$ values and after each iteration.
In particular, with $M=4$,  the ISS-MA shows remarkable performance gain (more than 5 dB gain).
Then, the performance gap decreases as $M$ increases.
% but the ISS-MA still maintains better performance.
%
%Note that the ISS-MA achieves faster BER improvement with iterations.
%
%
Note that the ISS-MA maintains strong performance even with small $M$  (e.g. $M=4$).
 Table~\ref{tb:comp1} provides computational complexity of both algorithms along with
the SNR required to achieve the BER of $10^{-2}$ for the same setup. The SNR is measured after the 7th iteration.
%
%It is shown that when $M$ is 4, the ISS-MA requires 11.26\% additional complexity over the $M$-algorithm.
%
In order to compare performance-complexity trade-off, it is worth looking at the performance of the ISS-MA with $M=4$ and the $M$-algorithm with $M=8$, where both algorithms require similar computational complexity.
In these cases, the ISS-MA achieves almost 1 dB performance gain.
%
%the complexity overhead of the ISS-MA is small  compared to the complexity of the ITS.
%
 We can additionally observe that the performance of the ISS-MA converges faster than the conventional $M$-algorithm, which might also help reducing
the complexity of the ISS-MA by the early termination of the iterations.

Next, we take a look at how the performance gap between the ISS-MA and the conventional $M$-algorithm changes in terms of different system size.
Table \ref{tb:comp2} presents the SNR at $10^{-3}$ BER and complexity of both algorithms for $N=L=6,8,10,$ and $12$.
 16-QAM is used and $N_l$ and $M$ are set to $5$ and $6$ for all cases. The performance is measured after the 7th iteration.
We observe that the performance gain due to the LE-LA path metric increases with system size.
In particular, the gain of the ISS-MA for $6\times6$ system is 0.5 dB and
that increases to 1.75 dB for the $12 \times 12$ system.
%
%Specifically, the ISS-MA achieves 1.75 dB gain for the $12 \times 12$ system while it achieves only 0.5 dB gain for the $6 \times 6$ system.
%
This clearly demonstrates that future cost plays a key role for large systems.

We next investigate the performance of the ISS-MA as a function of the parameter $N_l$ (see Section \ref{subsec:efficient}).
In our simulations, the $12 \times 12$ 16-QAM transmission is considered and $M$ is set to $8$.
Fig.~\ref{fig:nt} (a) and (b) show the performance and complexity of the ISS-MA for different $N_l$.
Note that  the ISS-MA with $N_l=0$ reduces to the conventional $M$-algorithm.
%
%
%Fig.~\ref{fig:nt} (b) provides the computational complexity needed for each $N_l$ value.
%
%
%
As the parameter $N_l$ increases, the ISS-MA accounts for the further future cost so that the computational complexity increases and performance improves.
  The ISS-MA offers performance-complexity trade-off through  $N_l$.
While the performance of the ISS-MA improves much  for small $N_l$ values,
the effect of $N_l$ diminishes with larger $N_l$.
It is shown that the choice of $N_l=5$ is sufficient to achieve the maximal performance gain offered by the LE-LA path metric for the cases considered.

We also take a look at the performance of the ISS-MA for spatially correlated MIMO channels. We model a correlated MIMO channel as $\mathbf{H}_{c} = R_{r}^{1/2} \cdot \mathbf{H} \cdot R_{t}^{1/2}$, where $R_{t}$ is the $N \times N$ transmit correlation matrix  and  $R_{r}$ is the $L \times L$ receive correlation matrix.
Fig. \ref{fig:correlated} shows the plot of BER vs. SNR of the ISS-MA and the conventional soft-input soft-output $M$-algorithm for correlated channels with
\begin{align}
R_t = R_r =  \begin{bmatrix}
1 & 0.8 & 0.8^2 & \cdots & 0.8^{12} \\
0.8 & 1 & 0.8 &  \cdots & 0.8^{11} \\
0.8^2 & 0.8 & 1 &  \cdots & 0.8^{10} \\
\vdots & \ddots & \ddots & \ddots & \vdots \\
0.8^{12} & 0.8^{11} & 0.8^{10} &  \cdots & 1 \\
\end{bmatrix}.
\end{align}
Note that two antennas are less correlated as the space between them increases.
The $12 \times 12$ 16-QAM system is considered and the parameters $M$ and $N_l$ are set to $12$ and $5$, respectively.
%
%Note that the antennas closer to each other lead to higher correlation.
%
%As shown in Fig. \ref{fig:correlated}, the performance gain of the ISS-MA remains for the correlated channels.
%
Comparing the result shown in Fig. \ref{fig:correlated}  with that in Fig. \ref{fig:bervsm} (d), we observe that the performance of both algorithms degrades in correlated channels, but the performance gain of the ISS-MA over the conventional algorithm is even larger.
From this observation, we can deduce that the LE-LA path metric can be more powerful when channel gains are correlated.

%%
%Hence, the complexity we provided here is those involved in tree search step and APP LLR computation step.
%%
%
%
%
%
%First, we take a look at the performance of MAP detectors.
%
%
%We compare the ext-$M$ algorithm with  the following algorithms
%%
%\begin{enumerate}
%\item Conventional $M$ algorithm ($M,K$) \cite{jong} - the traditional path metric was used in this method. To compare it with our modified path metric, all configurations except the path metric are set to be equal to those in our algorithm. Therefore, it is also characterized by two parameters $M$ and $K$.
%\item LISS algorithm ($|S|, |S_{x}|$) \cite{hagenauer} - this algorithm is based on a sequential tree search algorithm. It is characterized by the size of stack  $|S|$ and that of auxiliary stack $|S_{x}|$.
%\item  Hard sphere decoding ($i$) \cite{wang} - a best single candidate is found by hard sphere search based on SE enumeration. Then, the candidate list is generated by a bit-flipping method. This algorithm is characterized by $i$ which represents how many bits were flipped for list generation.
%\item Parallel sphere decoding \cite{jalden}, \cite{studer} - finds optimal max-log solution. The cost function based on APP probability used in \cite{jalden} and the bound tightening proposed in \cite{studer} are combined. An initial radius is set to infinity.
%\end{enumerate}

%\subsection{Performance Comparison}
Finally, we check the performance-complexity trade-off of the ISS-MA along with those of the existing soft-input soft-output tree detection algorithms.
For a comprehensive picture, we consider the following algorithms;
\begin{enumerate}
\item MMSE-PIC algorithm; MMSE parallel interference cancellation  detector \cite{singer,haykin}. This detector subtracts \emph{a priori} estimates of the interfering symbols from the
received vector and then applies a linear MMSE estimator to obtain soft estimate of the symbols.
\item LISS algorithm ($|S|, |S_{x}|$); List sequential stack algorithm \cite{hagenauer2007}. It is characterized by the size of stack  $|S|$ and that of auxiliary stack $|S_{x}|$.
%\item  Hard sphere decoding (HSD) \cite{giannakis2006} - After applying hard sphere search, a candidate list is generated by flipping each bit of the MAP estimate. Only single bit flipping is considered for the APP generation.
%\item Parallel sphere decoding \cite{jalden}, \cite{studer} - finds optimal max-log solution. The cost function based on APP probability used in \cite{jalden} and the bound tightening proposed in \cite{studer} are combined. An initial radius is set to infinity.
\item LFCSD algorithm ($N_{\mathcal{L}},N_{\mathcal{S}_{e}}$); List fixed complexity sphere decoder \cite{barbero2008}. A candidate list is found by the fixed complexity sphere decoder proposed in \cite{jalden2007}. This detector is characterized by $N_{\mathcal{L}}$ and $N_{\mathcal{S}_{e}}$, which represent the size of the candidate list and the number of paths fully extended, respectively.
\item ITS algorithm ($M$); Iterative tree search \cite{jong}. This detector uses the conventional $M$-algorithm to find the candidate list.
\end{enumerate}
Note that the parameters of the ISS-MA are remarked in ``ISS-MA ($M$, $N_l$)".
Although the LSD \cite{hochwald} and single (parallel) tree search (STS) \cite{jalden2005,studer} are considered as powerful detection
schemes, their complexities  grow so rapidly with problem size they are infeasible for the $12 \times 12$ system.
   For this reason, we only consider fixed-complexity detectors.
%
%While the ITS, LISS, and ISS-MA search for a candidate using \emph{a priori} information every iteration,
%%
%the LFCSD finds a candidate list only in the first iteration and use it in the subsequent iterations.
%Hence, the LFCSD needs larger candidate size.
%
%
%
In Fig.~\ref{fig:bervscomp}, the performance and complexity of each algorithm  are drawn in the same plot to compare the performance-complexity trade-off of the detectors.
Due to linear structure, the MMSE-PIC achieves the lowest complexity among all candidates.
In addition, the performance of the MMSE-PIC is better than that of the LISS, LFCSD, and ITS.
This would be why the performance of the latter detectors depends on  candidate size and  the size is not large enough to achieve good performance in the $12 \times 12$ system.
In particular, due to the limited stack size, the stack memory used in the LISS easily becomes full before reaching a leaf of the tree so that the LISS often fails to find reliable candidates.
%
%From the above results, in case when the smallest complexity is required, the MMSE-PIC could be a good candidate.
%
Fig.~\ref{fig:bervscomp} shows that only the ISS-MA can achieve the better performance than the MMSE-PIC.
Due to improved candidate selection process, the ISS-MA  finds reliable candidates only with small candidate size, thereby yielding the best BER performance while maintaining reasonable complexity.
In conclusion, the ISS-MA achieves the best performance-complexity trade-off among all tree detectors considered.
In addition, the ISS-MA provides performance gains over the MMSE-PIC at the expense of higher, but manageable complexity.

\section{Conclusions}
\label{sec:conclusion}
In this paper, we discussed a new path metric, which shows great promise in terms of its performance-complexity trade-off for soft-input soft-output tree detection in an IDD system.
By accounting for non-causal symbols in the \emph{linear estimate-based look-ahead} (LE-LA) path metric, the performance gains over the existing causal path metric are achieved.
We apply the LE-LA path metric to the soft-input soft-output $M$-algorithm.
By adopting the sorting mechanism exploiting the LE-LA path metric, we could improve the chance of selecting the correct path dramatically, thereby achieving good detection and decoding performance with fewer iterations.
From CPL probability analysis, we observed that the LE-LA path metric reflects the reliability of selected paths much better than the causal path metric.
Computer simulations confirm that the proposed ISS-MA can be a promising candidate for soft-input soft-output detection in high dimensional systems.
%
%Extensive study of the application of the LE-LA path metric to other tree search strategies needs to be done for future research.

%%%%%%%%%%%%%%%%%%%%%%%%%%%%%%%%%%%%%%%%%%%%%%%%%%%
%
%
%\clearpage

\appendices

\section{Proof of Theorem \ref{theorem:map}}
\label{app:map}
The transformed vector $\mathbf{y}$ can be expressed as $\mathbf{y} = \mathbf{R}\mathbf{x} + \mathbf{n}$, where  $\mathbf{n} = \mathbf{Q}_{1}\mathbf{n}_{o}$.
Let $k$ be the current layer being searched then
$\mathbf{y}$, $\mathbf{x}$, and
 $\mathbf{n}$ can be partitioned into two $(k-1) \times 1$ and $(N-k+1)\times1$ vectors, i.e.,
\begin{align}\label{eq:partition_appendix}
\mathbf{y} = \begin{bmatrix} \mathbf{y}_{1}^{k-1} \\ \mathbf{y}_{k}^{N} \end{bmatrix} = \begin{bmatrix}\mathbf{R}_{11,k} & \mathbf{R}_{12,k} \\ \mathbf{0} & \mathbf{R}_{22,k} \end{bmatrix}\begin{bmatrix} \mathbf{x}_{1}^{k-1} \\ \mathbf{x}_{k}^{N} \end{bmatrix} + \begin{bmatrix} \mathbf{n}_{1}^{k-1} \\ \mathbf{n}_{k}^{N} \end{bmatrix},
\end{align}
where the upper triangular matrix $\mathbf{R}$ is partitioned into four sub-matrices. Given the transmitted symbol $\mathbf{x}_{k}^{N}=\widetilde{\mathbf{x}}_{k}^{N}$, \emph{a posteriori} probability of $\mathbf{x}_{1}^{k-1}$  is given by
\begin{align} \label{eq:post1}
\ln & Pr\left(\mathbf{x}_{1}^{k-1} \Big| \mathbf{y}, \mathbf{x}_{k}^{N} = \widetilde{\mathbf{x}}_{k}^{N}\right) \\ =& \ln Pr\left(\mathbf{y}\Big|\mathbf{x}_{1}^{k-1}, \mathbf{x}_{k}^{N} = \widetilde{\mathbf{x}}_{k}^{N}\right) + \ln Pr(\mathbf{x}_{1}^{k-1})
\\   =&  -\ln\left(\sqrt{2\pi} \sigma_{n}\right) -\frac{1}{\sigma_{n}^{2}} \left\|\mathbf{y}  - \mathbf{R}\begin{bmatrix} \mathbf{x}_{1}^{k-1} \\ \widetilde{\mathbf{x}}_{k}^{N} \end{bmatrix}\right\|^{2} + \ln Pr\left(\mathbf{x}_{1}^{k-1}\right) \\
 =& -\ln\left(\sqrt{2\pi} \sigma_{n}\right)-\frac{1}{\sigma_{n}^{2}} \left\|
 \mathbf{y}_{1}^{k-1}  - \mathbf{R}_{12,k} \widetilde{\mathbf{x}}_{k}^{N} - \mathbf{R}_{11,k}\mathbf{x}_{1}^{k-1} \right\|^{2}   \label{eq:post_end}
   -\frac{1}{\sigma_{n}^{2}} \left\|\mathbf{y}_{k}^{N} - \mathbf{R}_{22,k} \widetilde{\mathbf{x}}_{k}^{N}\right\|^{2}    + \ln Pr\left(\mathbf{x}_{1}^{k-1}\right).
\end{align}
Hence, we can show that
\begin{align}
\check{\mathbf{x}}_{1}^{k-1} & = \arg \max \ln  Pr\left(\mathbf{x}_{1}^{k-1} \Big| \mathbf{y}, \mathbf{x}_{k}^{N} = \widetilde{\mathbf{x}}_{k}^{N}\right)\\
& = \arg \min_{\mathbf{x}_{1}^{k-1}} \left\|
 \mathbf{y}_{1}^{k-1} - \mathbf{R}_{12,k} \widetilde{\mathbf{x}}_{k}^{N}- \mathbf{R}_{11,k}\mathbf{x}_{1}^{k-1}  \right\|^{2} - \sigma_{n}^{2} \ln Pr\left(\mathbf{x}_{1}^{k-1}\right) \\ &= \arg \min_{\mathbf{x}_{1}^{k-1}}\sum_{i=1}^{k-1} \left( \left|y_{i} - \sum_{j=k}^{N}r_{i,j}\widetilde{x}_{j} - \sum_{j=i}^{k-1}r_{i,j}x_{j}  \right|^{2}\right. \left.- \sigma_{n}^{2}  \sum_{j=1}^{Q} \ln Pr\left(\widetilde{c}_{i,j}\right)\right)\\ \label{eq:proofmap}
 &= \arg \min_{\mathbf{x}_{1}^{k-1}} \sum_{i=1}^{k-1}  b\left(\mathbf{x}_{i}^{N}\right)
\end{align}
where the equation (\ref{eq:proofmap}) follows from the definition of the branch metric. Hence, for  $\mathbf{x}_{k}^{N} = \widetilde{\mathbf{x}}_{k}^{N}$, we have $\min_{\mathbf{x}_{1}^{k-1}} \sum_{i=1}^{k-1}  b\left(\mathbf{x}_{i}^{N}\right) = \sum_{i=1}^{k-1}  b\left(\mathbf{x}_{i}^{N}\right) \big|_{\mathbf{x}_{1}^{k-1} = \check{\mathbf{x}}_{1}^{k-1}} $.

\section{Proof of (\ref{eq:recursive})}
\label{app:recursive}
We can express $\mathbf{Z}_{k+1}$ in (\ref{eq:z2}) as
\begin{align}
& \mathbf{Z}_{k+1} = \sigma_{n}^{2}\left[\mathbf{R}_{11,k+1}\mathbf{\Lambda}_{k+1}(\mathbf{R}_{11,k+1})^{H} + \sigma_{n}^{2}\mathbf{I}\right]^{-1} \\
=& \sigma_{n}^{2} \left(\begin{bmatrix}\mathbf{R}_{11,k} & {\mathbf{r}}_{k+1}
\\ \mathbf{0} & r_{k+1,k+1}\end{bmatrix} \begin{bmatrix}\mathbf{\Lambda}_{k} & \mathbf{0}
\\ \mathbf{0} & \lambda_{k+1}\end{bmatrix}\begin{bmatrix}\mathbf{R}_{11,k} & {\mathbf{r}}_{k+1}
\\ \mathbf{0} & r_{k+1,k+1}\end{bmatrix}^{H}  + \begin{bmatrix} \sigma_{n}^{2}\mathbf{I}_{k} & \mathbf{0} \\ \mathbf{0} & \sigma_{n}^{2} \end{bmatrix}  \right)^{-1} \\ \label{eq:zbl3}
=& \sigma_{n}^{2} \left(\begin{bmatrix}\sigma_{n}^{2}\left(\mathbf{Z}_{k}\right)^{-1} + \lambda_{k+1}{\mathbf{r}}_{k+1}{\mathbf{r}}_{k+1}^{T} & \lambda_{k+1} r_{k+1,k+1} {\mathbf{r}}_{k+1} \\ \lambda_{k+1} r_{k+1,k+1} {\mathbf{r}}_{k+1}^{T} & \lambda_{k+1}r_{k+1,k+1}^{2} + \sigma_{n}^{2} \end{bmatrix}\right)^{-1}
\end{align}
To obtain the update formula, for partitioned matrices, $\mathbf{A}$ given by  \cite[\emph{Appendix 1.1.3}]{kai}
\begin{equation}
\mathbf{A} = \begin{bmatrix}\mathbf{A}_{11} & \mathbf{A}_{12} \\ \mathbf{A}_{21} & \mathbf{A}_{22} \end{bmatrix} \nonumber
\end{equation}
we have
%\begin{figure*}
\begin{align}
\mathbf{A}^{-1} = \begin{bmatrix}\left(\mathbf{A}_{11}-\mathbf{A}_{12}\mathbf{A}_{22}^{-1}\mathbf{A}_{21}\right)^{-1} & -
\left(\mathbf{A}_{11}-\mathbf{A}_{12}\mathbf{A}_{22}^{-1}\mathbf{A}_{21}\right)^{-1}\mathbf{A}_{12} \mathbf{A}_{22}^{-1}  \\ -\left(\mathbf{A}_{22}-\mathbf{A}_{21}\mathbf{A}_{11}^{-1}\mathbf{A}_{12}\right)^{-1}
\mathbf{A}_{21}\mathbf{A}_{11}^{-1} & \left(\mathbf{A}_{22}-\mathbf{A}_{21}\mathbf{A}_{11}^{-1}\mathbf{A}_{12}\right)^{-1} \end{bmatrix}. \nonumber
\end{align}
%\end{figure*}
Let $\mathbf{A}_{11} = \sigma_{n}^{2}\left(\mathbf{Z}_{k}\right)^{-1} + \lambda_{k+1}{\mathbf{r}}_{k+1}{\mathbf{r}}_{k+1}^{T}$, $\mathbf{A}_{12} =\lambda_{k+1} r_{k+1,k+1} {\mathbf{r}}_{k+1} $, $\mathbf{A}_{21} =\lambda_{k+1} r_{k+1,k+1} {\mathbf{r}}_{k+1}^{T} $, and $\mathbf{A}_{22} = \lambda_{k+1}r_{k+1,k+1}^{2} + \sigma_{n}^{2}$, then   (\ref{eq:zbl3}) becomes
\begin{align}
\mathbf{Z}_{k+1} = \begin{bmatrix}\mathbf{Z}_{k} - K\lambda_{k+1}\mathbf{Z}_{k}{\mathbf{r}}_{k+1}
{\mathbf{r}}_{k+1}^{H}\mathbf{Z}_{k} & -K\lambda_{k+1}r_{k+1,k+1}\mathbf{Z}_{k}{\mathbf{r}}_{k+1} \\ -K\lambda_{k+1}r_{k+1,k+1}{\mathbf{r}}_{k+1}^{H}\mathbf{Z}_{k} & K\left(\lambda_{k+1}{\mathbf{r}}_{k+1}^{H}\mathbf{Z}_{k}{\mathbf{r}}_{k+1} + \sigma_{n}^{2}\right) \end{bmatrix},
\end{align}
where
\begin{equation}
K = \frac{1}{\lambda_{k+1}{\mathbf{r}}_{k+1}^{T}\mathbf{Z}_{k}
{\mathbf{r}}_{k+1} + \lambda_{k+1}r_{k+1,k+1}^{2} + \sigma_{n}^{2}}.
\end{equation}

\section{Proof of Lemma \ref{lemma:sinr_bound}}
\label{app:sinr_bound}
Let $\mathbf{\Sigma}_{k}$ be decomposed to $\mathbf{U}\mathbf{\Phi}_{k}\mathbf{U}^{H}$, where $\phi_{1} \geq \phi_{2} \geq \cdots \geq \phi_{k}$ are the eigenvalues of $\mathbf{\Sigma}_{k}$.
Then, the upper bound of the SINR is given by
\begin{align}
 {\rm SINR} =&\frac{1}{\sigma_{n}^{2}}\frac{\left(\mathbf{r}_{k}^{H}\left(\sigma_{n}^{4} \mathbf{\Sigma}_{k}^{-2}\right)\mathbf{r}_{k} + |r_{k,k}|^{2}\right)^{2}}{\mathbf{r}_{k}^{H}\left(\sigma_{n}^{6} \mathbf{\Sigma}_{k}^{-3}\right)\mathbf{r}_{k} + |r_{k,k}|^{2}}
\\ = & \frac{1}{\sigma_{n}^{2}}\frac{\left(\sum_{i=1}^{k-1} \sigma_{n}^{4}  \phi_{i}^{-2} |r_{i,k}^{'}|^{2} + |r_{k,k}|^{2}\right)^{2}}{\sum_{i=1}^{k-1} \sigma_{n}^{6}  \phi_{i}^{-3} |r_{i,k}^{'}|^{2} + |r_{k,k}|^{2}} \\\label{eq:cauchy1}
\leq & \sum_{i=1}^{k-1}   \phi_{i}^{-1} |r_{i,k}^{'}|^{2} + \frac{|r_{k,k}|^{2}}{\sigma_{n}^{2}} \\ \label{eq:cauchy2}
= & \mathbf{r}_{k}^{H}\mathbf{\Sigma}_{k}^{-1} \mathbf{r}_{k}+ \frac{|r_{k,k}|^{2}}{\sigma_{n}^{2}},
\end{align}
where $\mathbf{r}'_{k} = \left[r'_{1,k}, \cdots, r'_{k-1,k} \right]^{T} = \mathbf{U} \mathbf{r}_{k}$ and  (\ref{eq:cauchy1}) is from the \emph{Cauchy-Schwarz} inequality.

Next, with $A = \mathbf{r}_{k}^{H}\left(\sigma_{n}^{4} \mathbf{\Sigma}_{k}^{-2}\right)\mathbf{r}_{k} + |r_{k,k}|^{2}$ and $ B= \mathbf{r}_{k}^{H}\left(\sigma_{n}^{6} \mathbf{\Sigma}_{k}^{-3}\right)\mathbf{r}_{k} + |r_{k,k}|^{2}$, we can show
\begin{align}
A - B &= \mathbf{r}_{k}^{H}\left(\sigma_{n}^{4} \mathbf{\Sigma}_{k}^{-2} - \sigma_{n}^{6} \mathbf{\Sigma}_{k}^{-3}\right)\mathbf{r}_{k} \\
& = \sum_{i=1}^{k-1} \left(\sigma_{n}^{4}  \phi_{i}^{-2} - \sigma_{n}^{6}  \phi_{i}^{-3}\right)|r_{i,k}^{'}|^{2} \geq 0 \label{eq:pospp}
\end{align}
where (\ref{eq:pospp}) follows from $\sigma_{n}^{2}  \phi_{i}^{-1} = {\sigma_{n}^{2}}/{(\sigma_{n}^{2} + \epsilon)}  \leq 1$. Hence,
\begin{equation}
{\rm SINR} = \frac{1}{\sigma_{n}^{2}}\frac{A^{2}}{B} \geq \frac{1}{\sigma_{n}^{2}}A = \sigma_{n}^{2}\mathbf{r}_{k}^{H} \mathbf{\Sigma}_{k}^{-2}\mathbf{r}_{k} + \frac{|r_{k,k}|^{2}}{\sigma_{n}^{2}}.
\end{equation}
This becomes the lower bound of the SINR.

\section{Proof of Theorem \ref{theorem:random}}
\label{app:random}
Let  $\mathbf{H}_{1:k-1}$ be a matrix generated from the first $k-1$ columns of $\mathbf{H}$. Since $\mathbf{H}_{1:k-1} = \mathbf{Q} \begin{bmatrix}\mathbf{R}_{11,k} \\ \mathbf{0} \end{bmatrix}$,
the matrices  $\mathbf{R}_{11,k}\mathbf{R}_{11,k}^{H}$ and $\mathbf{H}_{1:k-1}^{H} \mathbf{H}_{1:k-1}$ share same eigenvalues.
For an i.i.d. random matrix $\mathbf{H}$, the elements of $\mathbf{r}_{k}$ are zero-mean and independent with variance of $\frac{1}{L}$.  According to \cite[\emph{Lemma 2.29}]{verdu_mono},
as $N,L \rightarrow \infty$ with $\beta = {L}/{N}$, $B_k^{\rm upper}$ converges almost surely to
   \begin{align} \label{eq:vm}
      B_k^{\rm upper} = \mathbf{r}_{k}^{H}\left(\sigma_{n}^{2}\mathbf{I} + \lambda_{\rm min} \mathbf{R}_{11,k}\mathbf{R}_{11,k}^{H}\right)^{-1}\mathbf{r}_{k} \longrightarrow
      \frac{\gamma\beta}{\sigma_{n}^{2}} \int_{0}^{\infty} \frac{1}{1+\frac{\lambda_{\rm min}}{\sigma_{n}^{2}}x} f_{\eta}(x) dx
      %= \frac{\gamma\beta}{\lambda_{\rm min}}S_{\eta} \left(-\frac{\sigma_{n}^{2}}{\lambda_{\rm min}}\right),
      \end{align}
      where
      %$S_{x}(z)$ is Stieltjes transform of a random variable $x$, i.e., $S_{x}(z) = \int_{-\infty}^{\infty}\frac{1}{\lambda - z}f_{x}(\lambda) d\lambda$ and
      $f_{\eta}(x)$ is an empirical eigenvalue distribution of $\mathbf{H}_{1:k-1}^{H} \mathbf{H}_{1:k-1}$.
   According to the Marcenko-Pastur law \cite[{\it Theorem 2.35}]{verdu_mono}, as $N,L \rightarrow \infty$ with $\beta = {L}/{N}$, $f_{\eta}(x)$  converges almost surely to
   \begin{align} \label{eq:marcenko}
   %S_{\eta} \left(z\right) = \frac{1 - \gamma\beta - z - \sqrt{z^2 - 2(\gamma \beta +1)z + (\gamma \beta-1)^2} }{2\gamma\beta z},
   f_{\eta}(x) \longrightarrow f_{\eta}^{o}(x)  = \frac{\sqrt{\left(x-(1-\sqrt{\gamma\beta})^2\right)^{+}\left((1+\sqrt{\gamma\beta})^2-x\right)^{+}}}{2\pi \gamma\beta x},
   \end{align}
   where $(x)^{+} = \max(0, x)$.
     Hence, from (\ref{eq:vm}) and (\ref{eq:marcenko}), we obtain
        \begin{align}
     B_k^{\rm upper} \longrightarrow \frac{\gamma\beta}{\sigma_{n}^{2}} \int_{0}^{\infty} \frac{1}{1+\frac{\lambda_{\rm min}}{\sigma_{n}^{2}}x} f_{\eta}^{o}(x) dx =  \frac{1}{2\lambda_{\rm min}} \left(-1 - (1 - \gamma\beta)\frac{\lambda_{\rm min}}{\sigma_{n}^{2}}  + G\left(\frac{\lambda_{\rm min}}{\sigma_{n}^{2}}, \gamma\beta \right)\right).
      \end{align}

In a similar manner, the lower bound converges to
  \begin{align} \label{eq:vm2}
      B_k^{\rm lower} = \sigma_{n}^{2} \mathbf{r}_{k}^{H}(\sigma_{n}^{2}\mathbf{I} + \lambda_{\rm max} \mathbf{R}_{11,k}\mathbf{R}_{11,k}^{H})^{-2} \mathbf{r}_{k}  \longrightarrow &
      \frac{\gamma\beta}{\sigma_{n}^{2}} \int_{0}^{\infty} \frac{1}{\left(1+\frac{\lambda_{\rm max}}{\sigma_{n}^{2}}x\right)^{2}} f_{\eta}^{o}(x) dx  \\
      & = \frac{1}{2\sigma_{n}^{2}} \left(- \left(1 - \gamma\beta\right) + \frac{1+\gamma \beta + (1 - \gamma \beta)^{2} \frac{\lambda_{\rm max}}{\sigma_{n}^{2}}}{G\left(\frac{\lambda_{\rm max}}{\sigma_{n}^{2}}, \gamma\beta \right)}\right).
      \end{align}

\section{Proof of Lemma \ref{lemma:scaling}}
\label{app:scaling}
Let $\eta_{1}, \eta_{2}, \cdots, \eta_{k-1}$ be the unordered eigenvalues of $\mathbf{R}_{11,k}\mathbf{R}_{11,k}^{H}$.
The scaling gain in (\ref{eq:avg}) can be expressed as
\begin{align} \label{eq:sc1}
E_{\mathbf{H}}\left[ \exp\left(-\frac{K}{2}\sigma_{n}^{2}\mathbf{r}_{k}^{H}\mathbf{\Sigma}_{k}^{-2}\mathbf{r}_{k}
   \right)\right] & \leq E_{\mathbf{H}}\left[ \exp\left(-\frac{K}{2}\sigma_{n}^{2}\mathbf{r}_{k}^{H}\left(\sigma_{n}^{2}\mathbf{I} + \lambda_{\rm max} \mathbf{R}_{11,k}\mathbf{R}_{11,k}^{H}\right)^{-2}\mathbf{r}_{k}
   \right)\right] \\
   &= E_{\mathbf{H}} \left[ \prod_{i=1}^{k-1}\exp\left(-\frac{K}{2}   \frac{\sigma_{n}^{2}}{\left(\lambda_{\rm max}\eta_{i}+\sigma_{n}^{2}\right)^{2}} |r_{i,k}|^{2}\right)
    \right] \\\label{eq:sc3}
   &=  E_{\mathbf{R}_{11,k}} \left[ \prod_{i=1}^{k-1} E_{r_{i,k}}\left[\left.\exp\left(-\frac{K}{2}   \frac{\sigma_{n}^{2}}{\left(\lambda_{\rm max}\eta_{i}+\sigma_{n}^{2}\right)^{2}} |r_{i,k}|^{2}\right) \right| \mathbf{R}_{11,k}\right]
    \right]   \\\label{eq:sc4}
   &= E_{\mathbf{R}_{11,k}}   \left[\prod_{i=1}^{k-1} \frac{1}{1+\frac{K}{2}\frac{\sigma_{n}^{2}}{\left(\lambda_{\rm max}\eta_{i}+\sigma_{n}^{2}\right)^{2}}}\right],
   %& \leq \left(\frac{1}{K}\right)^{k-1}E_{\mathbf{R}_{11,k}}\left[{\rm det}\left(\mathbf{R}_{11,k}\mathbf{\Lambda}_{k} \mathbf{R}_{11,k}^{H} + \sigma_{n}^{2} \mathbf{I}\right)\right]
\end{align}
%The inequality (\ref{eq:sc1}) follows from that if $\mathbf{A} \preceq \mathbf{B}$, then $\mathbf{B}^{-1} \preceq \mathbf{A}^{-1}$ for $\mathbf{A} = \sigma_{n}^{2}\mathbf{I} + \mathbf{R}_{11,k}\mathbf{\Lambda}_{k}\mathbf{R}_{11,k}^{H}$ and $\mathbf{B} = \sigma_{n}^{2}\mathbf{I} + \lambda_{\rm max} \mathbf{R}_{11,k}\mathbf{R}_{11,k}^{H}$.
where (\ref{eq:sc3}) is from $E[x] = E[E[x|y]]$ and (\ref{eq:sc4}) follows from  $r_{i,k}$  being $\mathcal{CN}\left(0,1\right)$ and
independent of $\mathbf{R}_{11,k}$.
Let $\mathbf{H}_{1:k-1}$ be a matrix generated from the first $k-1$ columns of $\mathbf{H}$, then
the matrices  $\mathbf{R}_{11,k}\mathbf{R}_{11,k}^{H}$ and $\mathbf{H}_{1:k-1}^{H} \mathbf{H}_{1:k-1}$ share  same eigenvalues.
The pdf of the unordered eigenvalues of $\mathbf{H}_{1:k-1}^{H} \mathbf{H}_{1:k-1}$ is given by \cite{edelman}
\begin{align} \label{eq:jd}
f_{\eta_{1},  \cdots, \eta_{k-1}}\left(x_{1}, \cdots, x_{k-1}\right) = \frac{1}{(k-1)!}\exp\left(-\sum_{i=1}^{k-1} x_{i}\right) \prod_{i=1}^{k-1}\frac{x_{i}^{L-k+1}}{\left(k-1-i\right)!(L-i)!} \prod_{i<j}^{k-1} \left(x_{i} - x_{j}\right)^{2},
\end{align}
which completes the proof.
%Therefore, the upper-bound of the scaling factor is given by
%%
%\begin{align} \label{eq:scale_appendix}
%E_{\mathbf{H}}\left[ \exp\left(-\frac{K}{2}\mathbf{r}_{k}^{H}\mathbf{\Sigma}_{k}^{-1}\mathbf{r}_{k}
%   \right)\right] \leq \int \int_{0}^{\infty}  \prod_{i=1}^{k-1}  \frac{1}{1+\frac{1}{2}\frac{K}{\lambda_{\rm max}x_{i}+\sigma_{n}^{2}}}f_{\eta_{1},  \cdots, \eta_{k-1}}\left(x_{1}, \cdots, x_{k-1}\right) dx_{1} \cdots dx_{k-1}.
%\end{align}

\clearpage

\begin{figure} [!]
  %\centering
  \subfigure[]
  {\epsfig{figure=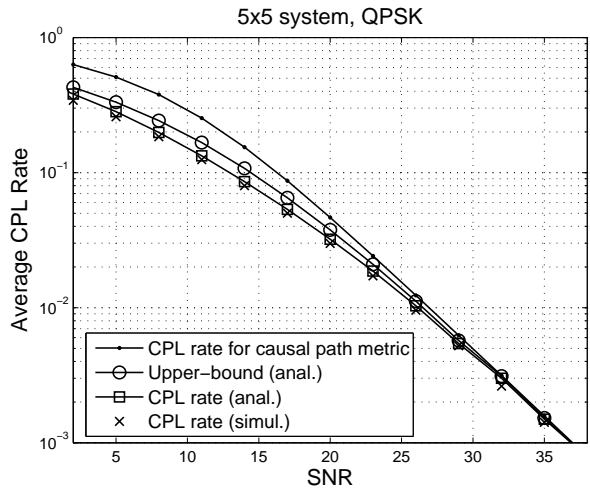,width=88mm}}
  \subfigure[]
  {\epsfig{figure=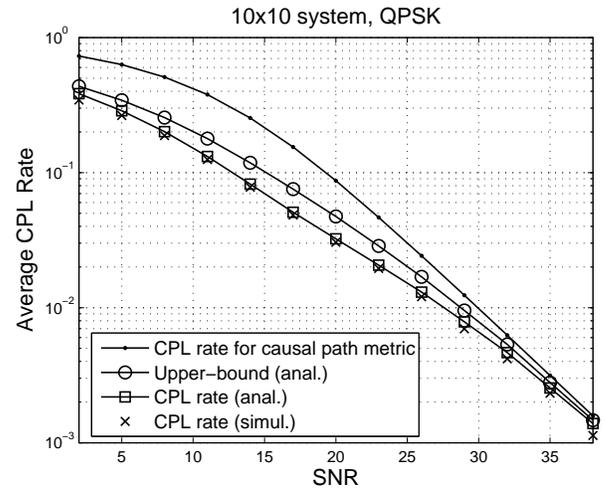,width=88mm}}
  \subfigure[]
  {\epsfig{figure=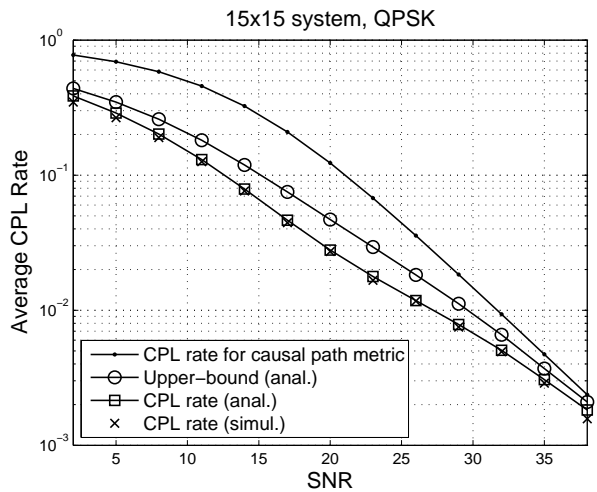,width=88mm}}
  \subfigure[]
  {\epsfig{figure=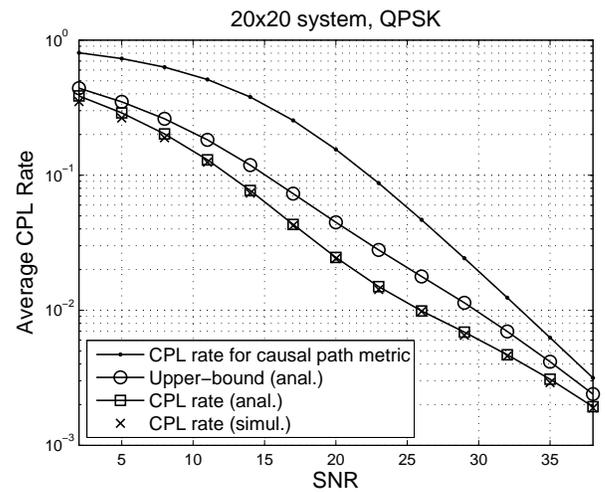,width=88mm}}
   \caption {Average CPL probability versus SNR for the (a) $5 \times 5$, (b) $10 \times 10$, (c) $15 \times 15$, and (d)  $20 \times 20$ systems. QPSK uncoded transmission is considered.
   The curves for the ``CPL rate (anal.)" are obtained by Monte-Carlo averaging of (\ref{eq:exact}) over i.i.d. Gaussian channels.
   } \label{fig:anal}
\end{figure}

\clearpage

\begin{figure}[!]
\centering
\centerline{\epsfig{figure=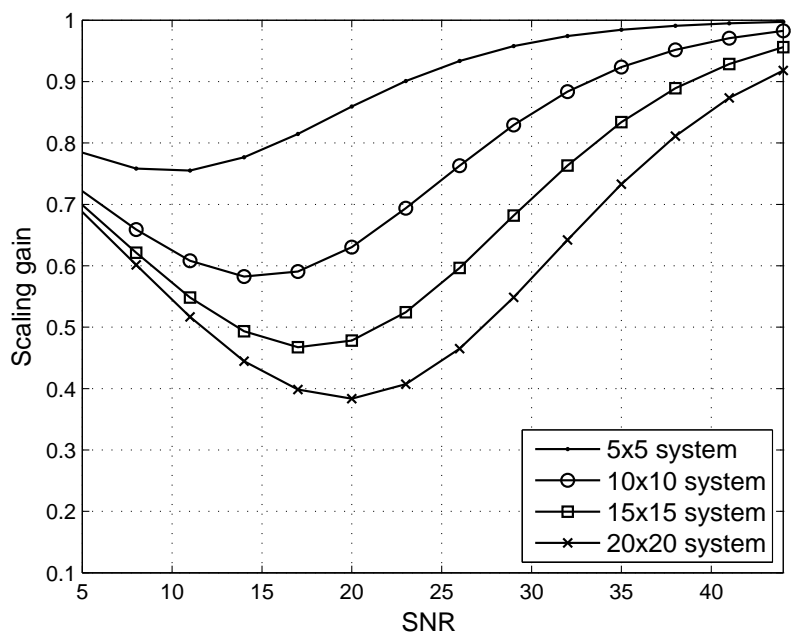,width=120mm}}
\caption{Scaling gain versus SNR for different system sizes $N=5,10,15,$ and $20$.} \label{fig:scale}
\end{figure}

\clearpage

\begin{figure} [!]
  \centering
  \subfigure[]
  {\epsfig{figure=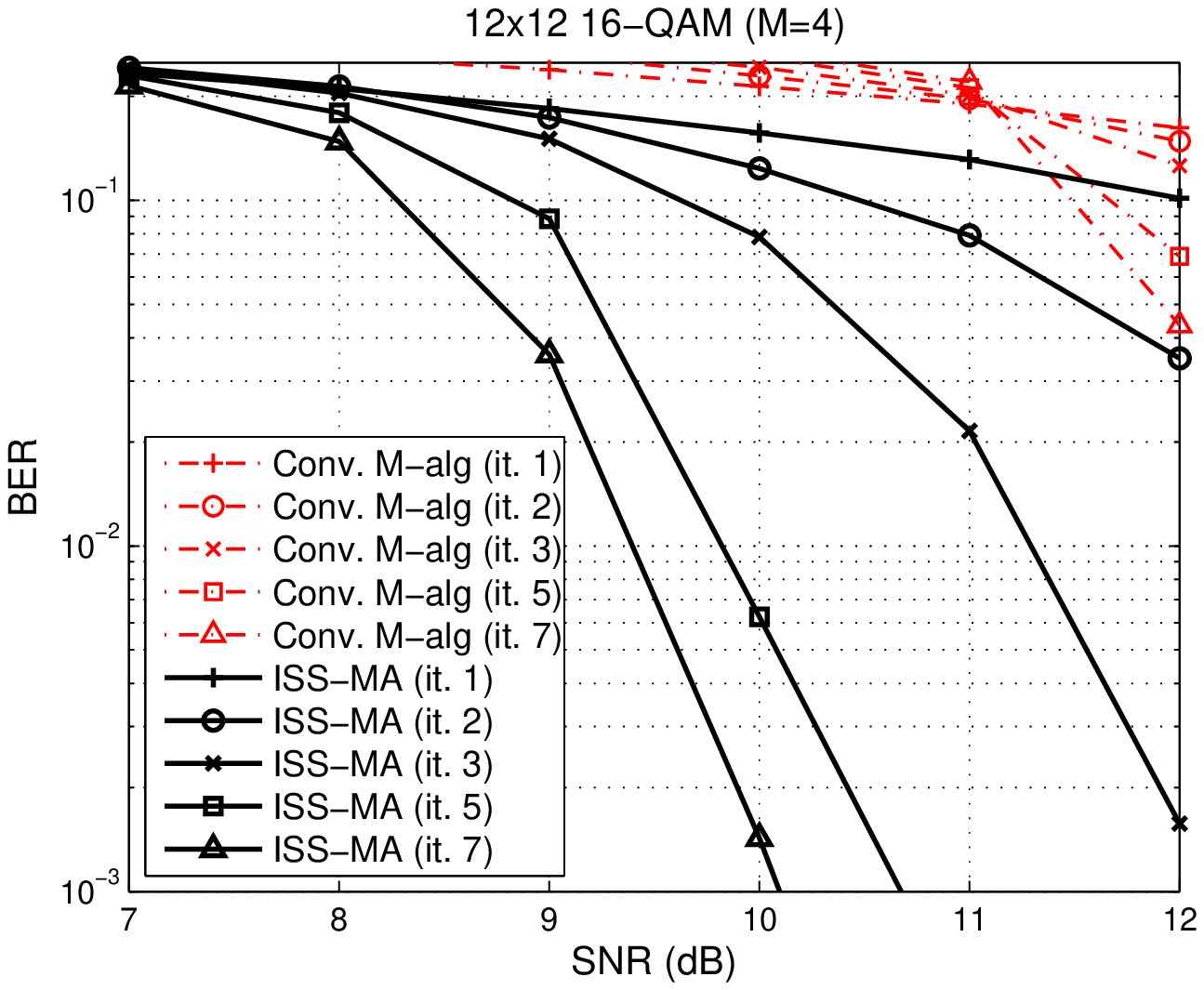,width=80mm}}
  \subfigure[]
  {\epsfig{figure=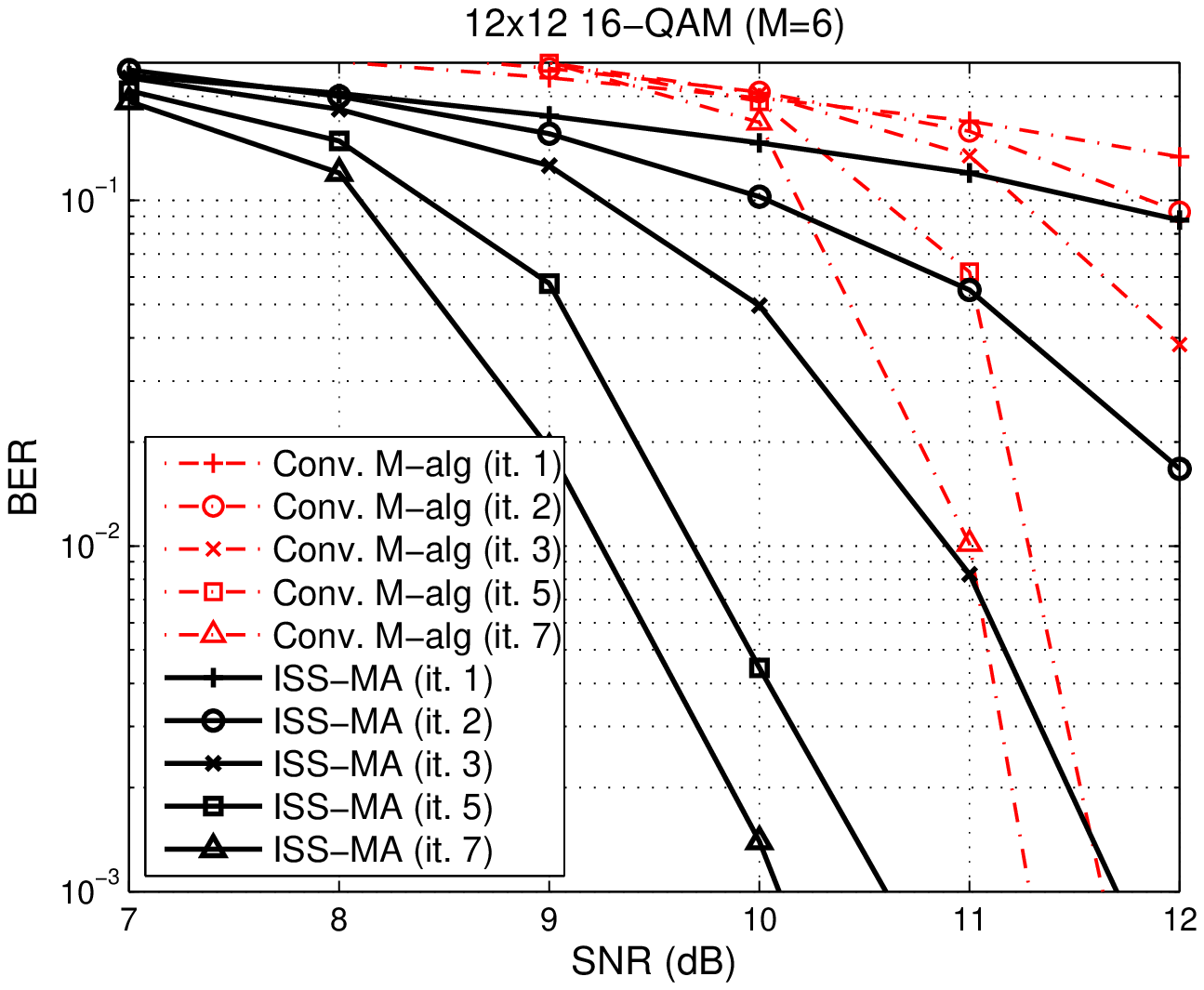,width=80mm}}
   \subfigure[]
  {\epsfig{figure=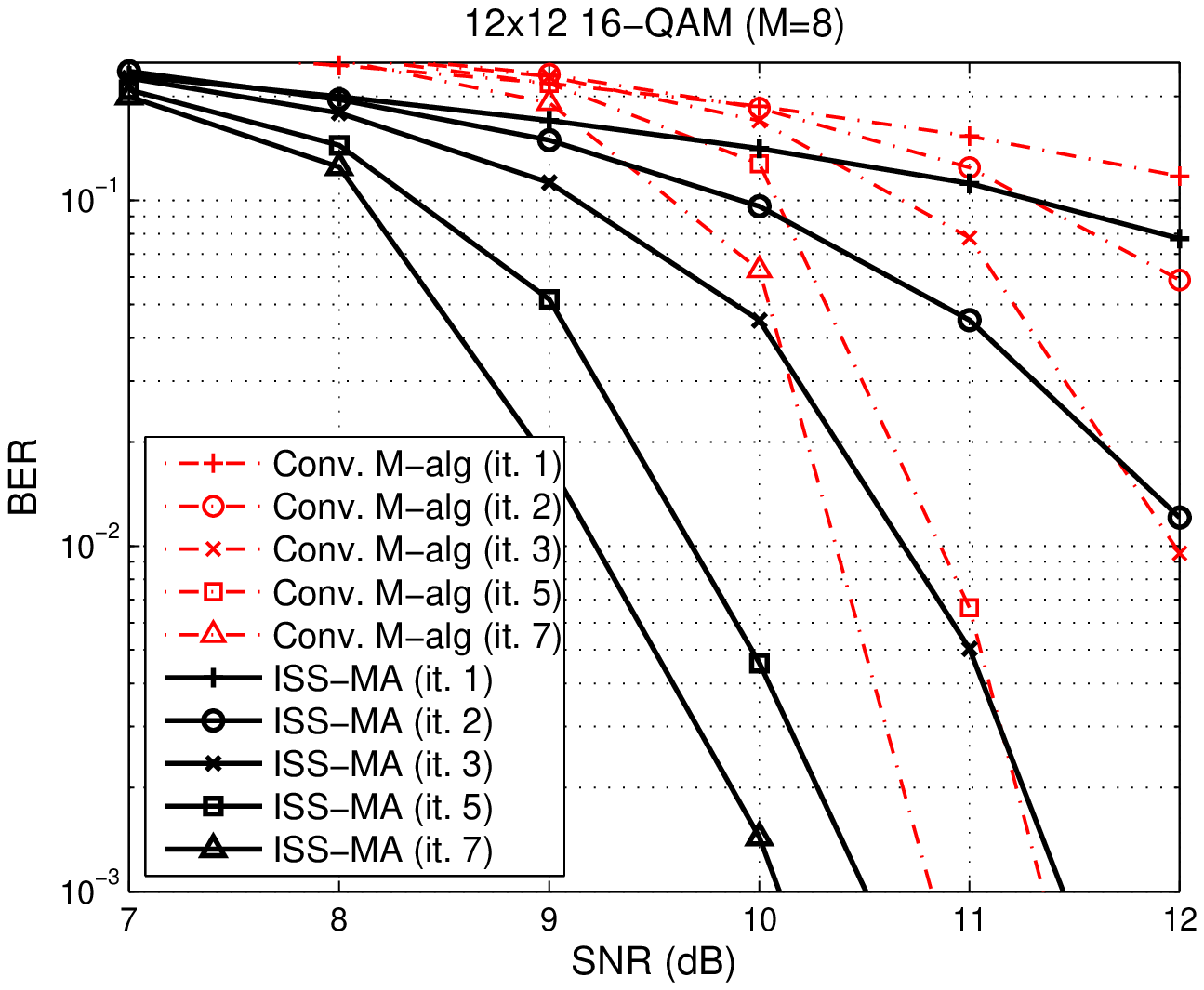,width=80mm}}
   \subfigure[]
  {\epsfig{figure=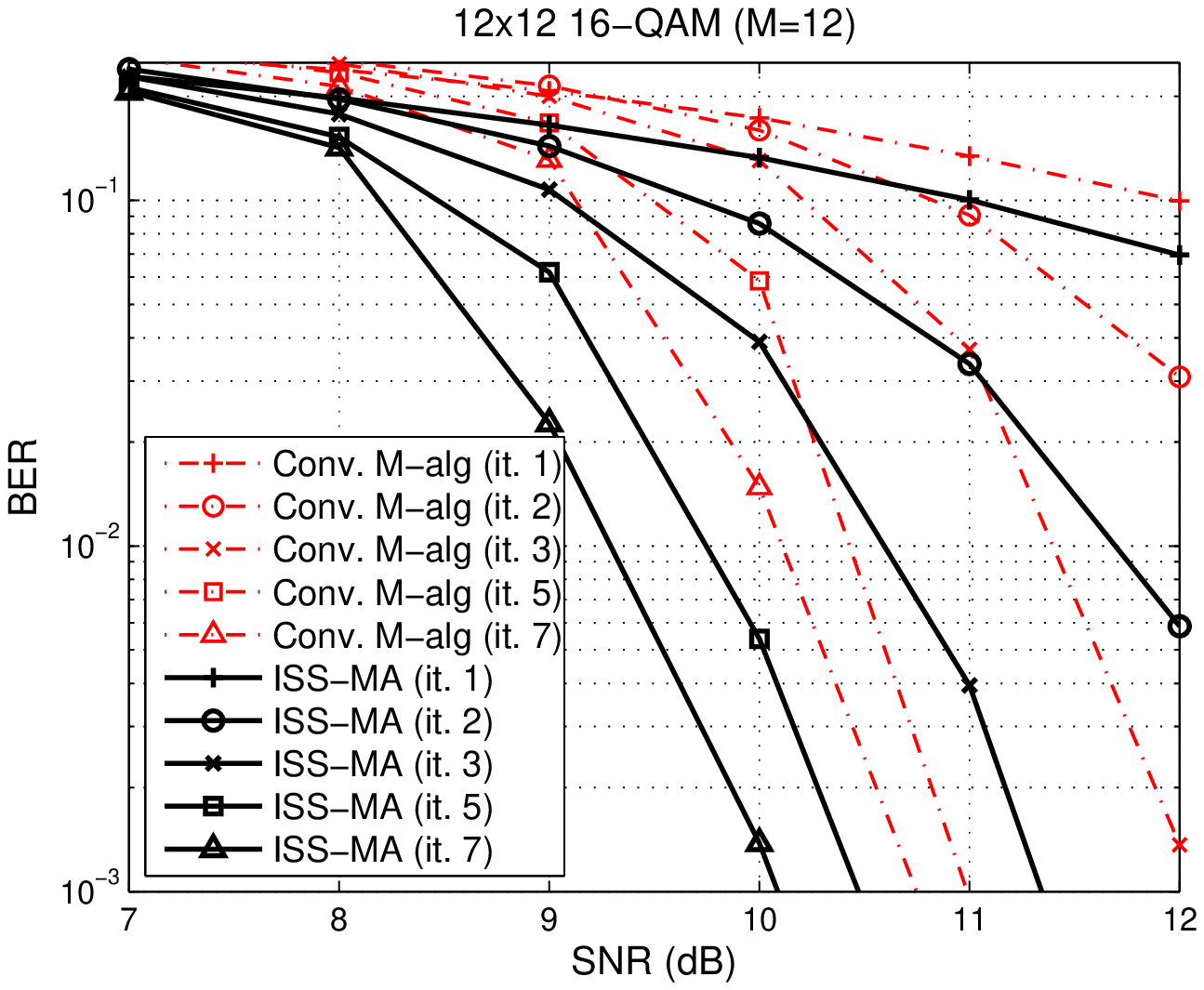,width=80mm}}
    \caption {Comparison between the causal path metric and LE-LA path metric for  the $12 \times 12$ 16-QAM system with
    (a) $M=4$, (b) $M=6$, (c) $M=8$, and (d) $M=12$.
   } \label{fig:bervsm}
\end{figure}

\clearpage
\begin{table*}[!]
\caption{Performance/complexity of $12 \times 12$ 16-QAM system for different $M$ values.}
\begin{center}
\resizebox{16cm}{!}{
\begin{tabular}{c||r|r|r|r}
\hline
 & \multicolumn{2}{|c|}{LE-LA path metric}  &  \multicolumn{2}{|c}{Conventional path metric}   \\ \hline
  & SNR  (at BER = 1\%)  & \# of multiplications  & SNR  (at BER = 1\%)   & \# of multiplications  \\ \hline
$M=4$ & \textbf{9.40} dB   & \textbf{133.77}k & 12.50 dB & 120.23k \\ \hline
$M=6$ & 9.25 dB  & 145.30k &  11.00 dB &  125.50k \\ \hline
$M=8$ & 9.22 dB  & 157.53k  & \textbf{10.36} dB  & \textbf{132.00}k \\ \hline
$M=12$ & 9.29 dB  & 183.12k  & 10.11 dB  & 145.98k \\ \hline
\end{tabular}
}
\end{center}
\label{tb:comp1}
\end{table*}

\begin{table*}[!]
\caption{Performance/complexity for different problem sizes. $M$ is set to 6.}
\begin{center}
\resizebox{16cm}{!}{
\begin{tabular}{c||r|r|r|r}
\hline
 & \multicolumn{2}{|c|}{LE-LA path metric}  &  \multicolumn{2}{|c}{Conventional path metric}   \\ \hline
  & SNR  (at BER = 1\%)  & \# of multiplications  & SNR  (at BER = 1\%)   & \# of multiplications  \\ \hline
$6\times 6$ 16-QAM & 8.80 dB   & 24.30k & 9.29 dB & 18.79k \\ \hline
$8\times 8$ 16-QAM & 8.97 dB  & 50.79k &  9.59 dB &  40.87k \\ \hline
$10\times 10$ 16-QAM & 9.22 dB  & 90.07k  & 10.39 dB  & 75.19k \\ \hline
$12\times 12$ 16-QAM & 9.25 dB  & 145.30k   & 10.11 dB  & 125.50k \\ \hline
\end{tabular}
}
\end{center}
\label{tb:comp2}
\end{table*}

%\begin{figure*}
%\begin{table*}[!]
%\caption{The number of complex multiplication per channel use and per iteration  for the ISS-MA and the  conventional (soft-input soft-output) $M$ algorithm.}
%\begin{center}
%\resizebox{14cm}{!}{
%\begin{tabular}{c||c|r|r}
%\hline
%System &   &  Complexity of conventional $M$-alg. &  Complexity of ISS-MA   \\ \hline
%\multirow{3}{*}{$20 \times 20$ QPSK}  & $M=12$   & 81k & 94k  \\ \cline{2-4}
% & $M=20$   & 97k & 115k  \\ \cline{2-4}
% & $M=56$  & 175k & 222k  \\ \hline
%\hline
%\multirow{3}{*}{$10 \times 10$ 16-QAM} & $M=12$   & 51k  & 72k \\ \cline{2-4}
% &  $M=20$  & 74k & 108k  \\ \cline{2-4}
% & $M=48$ & 156k  & 232k  \\ \hline
%\end{tabular}
%}
%\end{center}
%\label{tb:mult}
%\end{table*}
%%\end{figure*}
%
%\clearpage

\begin{figure} [!]
  \centering
  \subfigure[]
  {\epsfig{figure=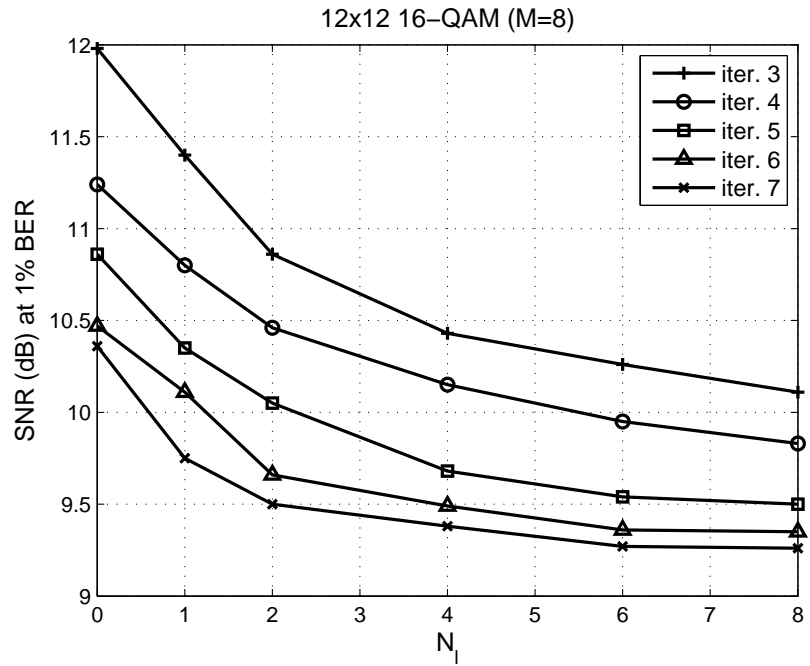,width=120mm}}
   \subfigure[]
  {\epsfig{figure=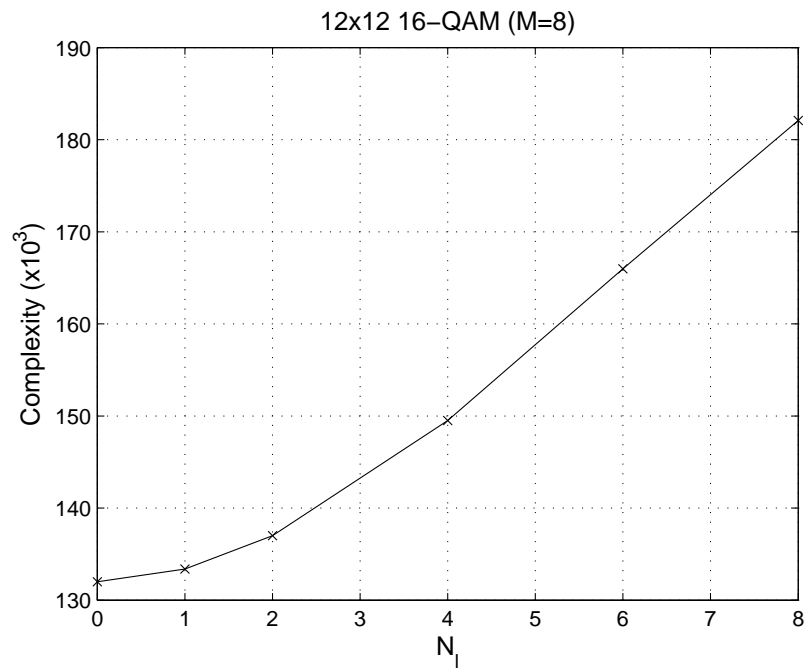,width=120mm}}
  \caption {(a) Performance and (b) complexity of the ISS-MA versus $N_l$ for $12 \times 12$ 16-QAM system ($M$ is set to 8).
   } \label{fig:nt}
\end{figure}

\clearpage

\begin{figure}[!]
\centering
\centerline{\epsfig{figure=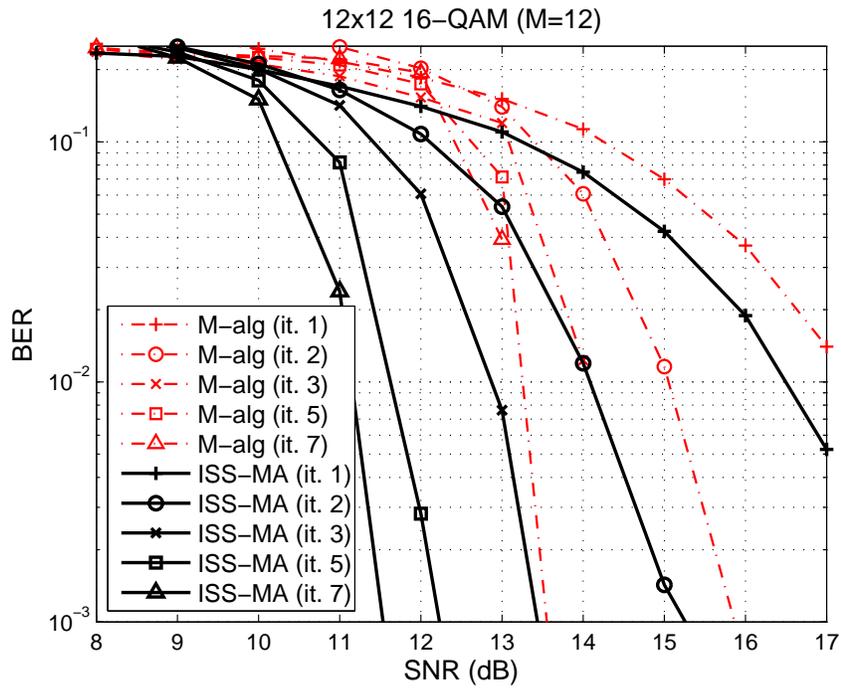,width=125mm}}
\caption{Comparison of the ISS-MA and conventional $M$-algorithm for spatially correlated MIMO channels.} \label{fig:correlated}
\end{figure}

\clearpage

\begin{figure}[!]
\centering
\centerline{\epsfig{figure=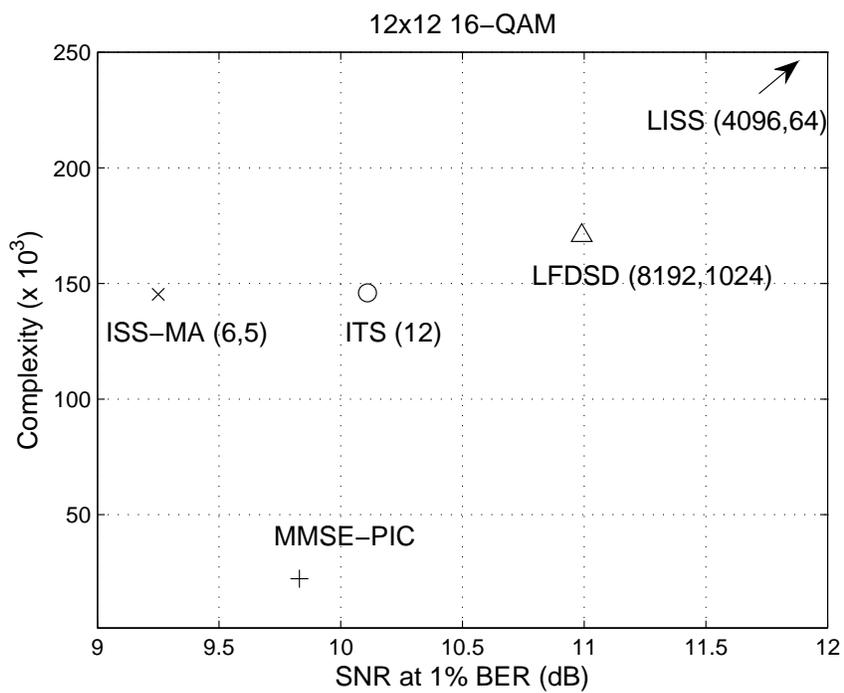,width=125mm}}
\caption{Comparison of the several soft-input soft-output detectors. The numbers in the parenthesis represent the parameters of
the detectors.} \label{fig:bervscomp}
\end{figure}

%\begin{table*}[!]
%\caption{Performance comparison of several soft-input soft-output tree detectors.}
%\begin{center}
%\resizebox{16cm}{!}{
%\begin{tabular}{c||r|r|r|r}
%\hline
% & \multicolumn{2}{|c|}{$20 \times 20$ QPSK}  &  \multicolumn{2}{|c}{$10 \times 10$ 16-QAM}   \\ \hline
%  & SNR thres.   & Complexity at SNR = $6$ dB & SNR thres.   & Complexity at SNR = $13$ dB  \\ \hline
%ISS-MA ($M=12$) & 4 dB   & 94k & 11.5 dB & 72k \\ \hline
%Conventional $M$-alg. ($M=12$) & 5 dB  & 81k &  12.0 dB &  51k \\ \hline
%LISS ($|S| = 4096$, $|S_{x}| = 64$) & 4.5 dB  & 1682k  & 12.5 dB  & 502k \\ \hline
%HSD   & 4 dB &  13222k &  11.5 dB & 6742k \\ \hline
%\end{tabular}
%}
%\end{center}
%\label{tb:comp}
%\end{table*}


\clearpage

\begin{thebibliography}{1}

%\bibitem{damen2000} M.\ O.\ Damen, A.\ Chkeif, and J. C. Belfiore,
%``Lattice code decoder for space-time codes," {\it IEEE Commun.
%Lett.}, vol. 4, pp. 161-163, May 2000.



%\bibitem{agrell} E.\ Agrell, T.\ Eriksson, A.\ Vardy, and K. Zegar,
%``Closet point search in lattices," {\it IEEE Trans. Information
%Theory}, vol. 48, pp. 2201-2214, Aug. 2002.
%
%\bibitem{viterbo1999} E.\ Viterbo and E.\ Giglieri, ``A universal
%lattice decoder," in {\it GRESTSI 14-eme Colloque}, Juan-les-Pins,
%France, Sep. 1993.




%
%\bibitem{hassibi2} B.\ Hassibi and H.\ Vikalo, ``On the
%sphere-decoding algorithm II. Generalizations, second-order
%statistics, and applications to communications," {\it IEEE Trans.
%Signal Processing}, vol. 53, pp. 2819-2834, Aug. 2005.
%
%\bibitem{jalden} J.\ Jalden and B.\ Ottersten, ``On the complexity
%of sphere decoding in digital communication," {\it IEEE Trans.
%Signal Proceesing}, vol. 53, pp. 1474-1484, April 2005.



\bibitem{singer} M.\ T\"{u}chler, R. Koetter, and A. C. Singer,
``Turbo equalization: principles and new results," {\it IEEE Trans.
Commun.}, vol. 50, pp. 754-767, May 2002.





\bibitem{haykin} M.\ Sellathurai and S.\ Haykin, ``Turbo-BLAST for
wireless communications: theory and experiments," {\it IEEE Trans.
Signal Processing}, vol. 50, pp. 2538-2546, Oct. 2002.


\bibitem{hochwald} B.\ Hochwald and S.\ T.\ Brink, ``Achieving
near-capacity on a multiple-antenna channel," {\it IEEE Trans.
Commun.}, vol. 51, pp. 389-399, March 2003.




\bibitem{vikalo2004} H.\ Vikalo, B.\ Hassibi, and T.\ Kailath, ``Iterative decoding for MIMO channels via modified sphere decoding," {\it IEEE Trans. Wireless Commun.}, vol. 3, pp. 2299-2311, Nov. 2004.


\bibitem{wang_poor} X.\ Wang and H.\ V.\ Poor, ``Iterative (turbo) soft
interference cancellation and decoding for coded CDMA," {\it IEEE
Trans. Commun.}, vol. 47, pp. 1046-1061, July 1999.




\bibitem{berrou} C.\ Berrou and A.\ Glavieux
``Near optimum error-correcting coding and decoding: Turbo-codes,"
{\it IEEE Trans. Commun.}, vol. 44, pp 1261-1271, Oct. 1996.

\bibitem{logmapsova} P.\ Robertson, P.\ Hoeher, and E. Villebrun,
``Optimal and sub-optimal maximum a posteriori algorithms suitable
for turbo decoding," {\it European Trans. on Telecommun.}, vol. 8,
pp. 119-125, March 1997.



\bibitem{giannakis2006} R. Wang and G. B. Giannakis, ``Approaching MIMO channel capacity with soft detection
based on hard sphere decoding," {\it IEEE Trans. Commun.}, vol. 54, pp. 587-590, April 2006.

\bibitem{jalden2005} J. Jald\'{e}n and B. Ottersten, ``Parallel implementation of a soft output sphere decoder,"
{\it Proc. IEEE Ashilomar Conference on Signals, Systems, and Computers}, Nov. 2005, pp. 581-585.

\bibitem{studer} C. Studer, A. Burg, and H. B\"{o}lcskei, ``Soft-output sphere decoding: algorithms and VLSI implementation," {\it IEEE Journal on Selected Areas in Commun.}, vol. 26, pp. 290-300, Feb. 2008.



\bibitem{jong} Y.\ L.\ C.\ de Jong and T. J. Wilink, ``Iterative tree search detection for MIMO wireless systems," {\it IEEE Trans. Commun.}, vol. 53, pp. 930-935, June 2005.

    \bibitem{milliner} D. L. Milliner, E. Zimmermann, J. R. Barry, and G. Fettweis, ``A fixed-compelxity smart candidate adding algorithm for soft output MIMO detection," {\it IEEE Journal of Selected Topics in Signal Processing}, vol. 3, pp. 1016-1025, Dec. 2009.

\bibitem{hagenauer2007} J. Hagenauer and C. Kuhn, ``The list-sequential (LISS) algorithm and its application," {\it IEEE Trans. Commun.}, vol. 55, pp. 918-928, May 2007.

\bibitem{barbero2008} L. G. Barbero and T. S. Thompson, ``Extending a fixed-complexity sphere decoder to obtain likelihood information for turbo-MIMO systems," {\it IEEE trans. Veh. Technol.}, vol. 57, no. 5, pp. 2804-2814, Sep. 2008.


\bibitem{larsson} E.\ G.\ Larsson and J.\ Jald\'{e}n, ``Fixed-complexity soft MIMO detection via partial marginalization," {\it IEEE Trans. Signal Processing}, vol. 56, pp. 3397-3407, Aug. 2008.

\bibitem{wu} D. Wu, J. Eilert, R. Asghar and D. Liu, ``VLSI implementation of a fixed-complexity soft-output MIMO detector for high-speed wireless,"
{\it EURASIP Journal on Wireless Commun. and Networking}, vol. 2010, pp. 1-13, 2010.

\bibitem{arul} A.\ D.\ Murugan, H.\ E.\ Gamal, M.\ O.\ Damen, and G.
Caire, ``A unified framework for tree search decoding: rediscovering
the sequential decoder," {\it IEEE Trans. Information Theory}, vol.
52, pp. 933-953, March 2006.

\bibitem{pohst} U.\ Fincke and M.\ Pohst, ``Improved methods for
calculating vectors of short length in a lattice, including a
complexity analysis," {\it Math. Comput.}, vol. 44, pp. 463-471,
April 1985.

%\bibitem{damen} M.\ O.\ Damen, H.\ E.\ Gamal, and G.\ Caire, ``On
%maximum-likelihood detection and the search for the closest lattice
%point," {\it IEEE Trans. Information Theory}, vol. 49, pp.
%2389-2402, Oct. 2003.

\bibitem{hassibi1} B.\ Hassibi and H.\ Vikalo, ``On the
sphere-decoding algorithm I. Expected complexity," {\it IEEE Trans.
Signal Processing}, vol. 53, pp. 2806-2818, Aug. 2005.


 \bibitem{jalden2007} J. Jald\'{e}n, L. G. Babero, B. Ottersten, and J. S. Thompson, ``Full diversity detection in MIMO systems with
 a fixed-complexity sphere decoder," {\it IEEE International Conference on Acoustics, Speech, and Signal Processing}, April 2007, pp. III-49-III-52.




\bibitem{malg} J.\ B.\ Anderson and S.\ Mohan, ``Sequential coding algorithms: A
survey and cost analysis," {\it IEEE Trans. Commun.}, vol. COM-32,
no. 2, pp. 169-176, Feb. 1984.



\bibitem{wong} K.\ Wong, C.\ Tsui, R.\ S.\ Cheng, and W. Mow, ``A VLSI
architecture of the K-best lattice decoding algorithm for MIMO channels,"
{\it IEEE ISCAS, Scottsdale, AZ, USA}, May 2002, pp. 273-276.

\bibitem{guo} Z.\ Guo and P.\ Nilsson, ``Algorithm and implementation of the K-best sphere
decoding for MIMO detection," {\it IEEE Journal on Selected Areas in Commun.}, vol. 24, pp. 491-503, March 2006.

\bibitem{vblast} P.\ W.\ Wolniansky, G.\ J.\ Foschini, G.\ D.\
Golden, and R.\ A.\ Valenzuela, ``V-BLAST: an architecture for
realizaing very high data rates over the rich-scattering wireless
channel," {\it Proc. URSI Int. Symp. Signals, Syst., Electron.}, Sep. 1998, pp.
295-300.




\bibitem{ai} J. N. Nilsson, {\it Principle of Artificial Intelligence.} Palo Alto, CA: Tioga Publishing Co., 1980.

\bibitem{hh} Y. S. Han, C. R. P. Hartmann, and C. Chen, ``Efficient priority-first search maximum-likelihood soft-decision decoding of linear block codes," {\it IEEE Trans. Information Theory,} vol. 39, pp. 1514-1523, Sep. 1993.

\bibitem{ekroot} L. Ekroot and S. Dolinar, ``A* decoding of block codes," {\it IEEE Trans. Commun.}, vol. 44, pp. 1052-1056, Sep. 1996.

\bibitem{stojnic} M. Stojnic, H. Vikalo, and B. Hassibi, ``Speeding up the sphere decoder with $H^{\infty}$ and SDP inspired lower bounds,"
{\it IEEE Trans. Signal Processing,}, vol. 56, pp. 712-726, Feb. 2008.

\bibitem{fano} R. Johannesson and K. S. Zigangirov, {\it Fundamentals of convolutional coding}, Wiley-IEEE Press, 1999.
%R. Fano, ``A heuristic discussion of probabilistic decoding," {\it IEEE Trans. Information Theory,} vol. 55, pp. 918-928, May 2007.

\bibitem{xiong} F.\ Xiong, A.\ Zerik, and E.\ Shwedyk, ``Sequentional sequence estimation for channels with intersymbol interference of finite or infinite length," {\it IEEE Trans.
Commun.}, vol. 38, pp. 795-804, June 1990.

\bibitem{gowaikar} R.\ Gowaikar and B.\ Hassibi, ``Statistical pruning for near-Maximum Likelihood Decoding,"
{\it IEEE Trans. Signal Processing}, vol. 55, pp. 2661-2675, June 2007

\bibitem{shim} B.\ Shim and I.\ Kang, ``Sphere
decoding with a probabilistic tree pruning," {\it IEEE Trans. Signal Processing},
vol. 56, pp. 4867-4878, Oct. 2008.

\bibitem{ho} T. Cui, T. Ho and C. Tellambura, ``Statistical pruning for near maximum likelihood detection of MIMO systems,"
{\it IEEE International Conference on Commun. (ICC)}, June 2007, pp. 5462-5467.

\bibitem{wanlun2006} W.\ Zhao and G.\ B.\ Giannakis, ``Reduced
complexity closest point decoding algorithms for random lattices,"
{\it IEEE Trans. Wireless Commun.}, vol. 5, pp. 101-111, Jan. 2006.



\bibitem{poor_book}  H.\ V.\ Poor, {\it An Introduction to Signal
Detection and Estimation, 2nd Edition}, Springer, 1994.

\bibitem{kai} S. M. Kay, {\it Fundamentals of statistical signal processing: estimation theory.}  Addison Wesley Longman, 1993.



\bibitem{singer1} M. T\"{u}chler, A. C. Singer, and R. Koetter, ``Minimum mean squared error equalization using \emph{a priori} information," {\it IEEE Trans. Signal Processing}, vol. 50, pp. 673-683, March 2002.

  \bibitem{clip} D. Milliner, E. Zimmermann, J. R. Barry, G. Fettweis, ``Channel state information based LLR clipping in list MIMO detection," {\it Proc. IEEE PIMRC}, Sept. 2008, pp. 15-18.




%\bibitem{zimmerman} E.\ Zimmermann, and G.\ Fettweis, ``Improved length term calculation and MMSE extension for LISS MIMO detection," {\it Proc. IEEE ISIT Workshop}, Oct. 2006, pp. 650-654.

%\bibitem{overlineo} S. overlineo, J. Hagenauer, and M. Witzke, ``Iterative detection of MIMO transmission using a list-sequential (LISS) detector," {\it Proc. IEEE ICC}, May 2003, Anchorge, USA.








%\bibitem{pham} D. Pham, K. R. Pattipati, P. K. Willett, and J. Luo, ``An improved complex
%sphere decoder for V-BLAST systems," {\it IEEE Signal Processing Letters}, vol. 11, pp. 748-751, Sep. 2004.



%\bibitem{ferrari} G. Colavolpe, G. Ferrari, and R. Raheli, ``Extrinsic information in iterative decoding: a unified view," {\it IEEE Trans. Commun.}, vol. 49, pp. 2088-2094, Dec. 2001.

%\bibitem{huang} C. X. Huang and A. Ghrayeb, ``A simple remedy for the exaggerated extrinsic information produced by the SOVA algorithm," {\it IEEE Trans. Wireless Commun.}, vol. 5, pp. 996-1002, May 2006.

    \bibitem{poor2} H. V. Poor and S. Verdu, ``Probability of error in MMSE multiuser detection," {\it IEEE Trans. Information Theory}, vol. 43, pp. 858-871, May 1997.
\bibitem{ping} P.\ Li, D. Paul, R. Narasimhan, and J. Cioffi, ``On the distribution of SINR for the MMSE MIMO receiver and performance analysis," {\it IEEE Trans. Information Theory}, vol. 52, pp. 271-286, Jan. 2006.

    \bibitem{dguo} D. Guo, S. Verdu, and L. K. Rasmussen, ``Asymptotic normality of linear multiuser receiver outputs," {\it IEEE Trans. Information Theory}, vol. 48, pp. 3080-3095, Dec. 2002.

\bibitem{proakis} J. Proakis, {\it Digital communications: 4th edition.} McGraw-Hill, 2001.

       \bibitem{verdu_mono} A. M. Tulino and S. Verdu, {\it Random matrix theory and wireless communications.} Foundations and Trends in Communications and Information Theory, 2004.

           \bibitem{viswanath} D.\ Tse and P.\ Viswanath, {\it Fundamentals of wireless communication.} Cambridge University Press, 2005.

           \bibitem{brink} S.\ T.\ Brink, ``Convergence of iterative decoding," {\it Electron. Lett.}, vol. 35, pp. 806-808, May 1999.


           \bibitem{edelman} A. Edelman, ``Eigenvalues and condition numbers of random matrices," Ph. D. thesis, Dept. Math., Massachusetts Inst. Technol., Cambridge, 1989.

   %
%    \bibitem{bittner} S.\ Bittner, E.\ Zimmermann, W.\ Rave, and G.\ Fettweis, ``List sequential MIMO detection: noise bias term and partial path augmentation," {\it IEEE International Conference on Commun. (ICC)}, June 2006,
%pp. 1300-1305.
%



















\end{thebibliography}
\end{document}